\documentclass[pre,twocolumn,usletter,superscriptaddress]{revtex4}

\usepackage{amsmath}
\usepackage{amsfonts}
\usepackage{amssymb}
\usepackage{amsthm}
\usepackage{color}
\usepackage{float}

\usepackage [ansinew] {inputenc}
\usepackage{graphicx}
\usepackage{natbib}

\theoremstyle{plain} 
\theoremstyle{plain} 
\theoremstyle{plain} \newtheorem{theor}{Theorem}[section] 
\theoremstyle{plain} 
\theoremstyle{plain} \newtheorem*{corol}{Corollary} 
\theoremstyle{remark} 
\theoremstyle{plain} 
\theoremstyle{remark} 

\begin{document}
\title{A new approach to fuzzy sets:\\ Application to the design of nonlinear time-series,\\ symmetry-breaking patterns, \\ and non-sinusoidal limit-cycle oscillations}

\author{Vladimir Garc\'{\i}a-Morales}
%\email{vmorales@ph.tum.de}

\affiliation{Departament de Termodin\`amica, Universitat de Val\`encia, E-46100 Burjassot, Spain}
\email{garmovla@uv.es}

%\affiliation{Nonequilibrium Chemical Physics - Physics Department - Technische Universit\"{a}t M\"{u}nchen, James-Franck-Str. 1, D-85748 Garching, Germany}

%\affiliation{Institute for Advanced Study - Technische Universit\"{a}t M\"{u}nchen, Lichtenbergstr. 2a, D-85748 Garching, Germany}

%\affiliation{Nonequilibrium Chemical Physics - Physics Department - Technische Universit\"{a}t M\"{u}nchen, James-Franck-Str. 1, D-85748 Garching, Germany}

\begin{abstract}
\noindent It is shown that characteristic functions of sets can be made fuzzy by means of the $\mathcal{B}_{\kappa}$-function, recently introduced by the author, where the fuzziness parameter $\kappa \in \mathbb{R}$ controls how much a fuzzy set deviates from the crisp set obtained in the limit $\kappa \to 0$.  As applications, we present first a general expression for a switching function that may be of interest in electrical engineering and in the design of nonlinear time-series. We then introduce another general expression that allows wallpaper and frieze patterns for every possible planar symmetry group (besides patterns typical of quasicrystals) to be designed. We show how the fuzziness parameter $\kappa$ plays an analogous role to temperature in physical applications and may be used to break the symmetry of spatial patterns. As a further, important application, we establish a theorem on the shaping of limit cycle oscillations far from bifurcations in smooth deterministic nonlinear dynamical systems governed by differential equations. Following this application, we briefly discuss a generalization of the Stuart-Landau equation to non-sinusoidal oscillators obtained as a consequence of our theorem. ~\\

\noindent \emph{Keywords:} Fuzzy sets; fuzzy logic; switching function; wallpaper groups; symmetry breaking; time series; limit cycles; nonlinearity
 \end{abstract}
%\pacs{05.45.Df, 03.65.Fd, 05.20.-y, 89.75.Fb}
\maketitle

\section{Introduction}
%%%% Insert A head here%\section{The inclusion-exclusion principle}To illustrate this method we first introduce some elementary concepts. 

The theory of fuzzy sets \cite{Zadeh,Klaua,Goguen} generalizes the traditional concept of a class (i.e. a set of objects for which a certain property holds) to situations where membership to the class is vaguely defined so that a continuum of grades of membership is possible. Technically speaking, the membership function of a fuzzy set (i.e. the characteristic or indicator function) is not a Boolean function as is the case of a `crisp' set (see \cite{ZimmermannREV} for a review). Rather, the truth value can be any real number in the unit interval. Fuzzy sets, and the closely associated concept of fuzzy logic, have been applied to a wide variety of fields including control theory and artificial intelligence \cite{Zimmermann}. 

In this article a new approach to fuzzy sets is presented. Our approach makes use of the $\mathcal{B}$-function \cite{VGM1} that, as we show here, allows characteristic functions of crisp sets to be specified. These characteristic functions become then fuzzyfied by means of the $\mathcal{B}_{\kappa}$-function \cite{JPHYSA,homotopon},  the one-parameter $\kappa$ deformation of the $\mathcal{B}$-function. The laws of algebra of both crisp and fuzzy sets are addressed. Our formulation is different to previous ones found in the literature (see Table 1 in \cite{ZimmermannREV}) because the `aggregation' operators (defining union and intersection of fuzzy sets) are different in our case.

%We then introduce the 16 binary fuzzy logical operators fuzzy $t$-norms and co-norms, also establishing a connection with Zadeh's fuzzy set theory. However, the fuzzification method here presented is new and may be potentially useful in applications to the physical sciences, artificial intelligence and control theory. 

We then present several applications of our theory to the mathematical modelling of nonlinear systems.  First, with help of the  $\mathcal{B}$-function and its fuzzy version we formulate a general expression for a switching function that may be of interest in electrical engineering \cite{Marouchos} and in the design of nonlinear time-series. It contrasts with the ones found in the literature \cite{Marouchos} in that our switching function can be tuned from sinusoidal to square-wave-like by changing the fuzziness parameter $\kappa$.

% Furthermore there is no need to take any number of Fourier modes to approximate the square wave and there is no associated Gibbs phenomenon.
%First, with the methods presented here, we sketch how model functions that are both continuous and differentiable, can be constructed with prescribed asymptotic properties. We show an application of this to the mathematical modelling of nonlinear current-voltage characteristics of non-ideal diode elements of electrical circuits.

We then introduce a single general expression that allows to produce wallpaper and frieze patterns with all possible symmetry groups, as well as patterns typical of quasicrystals \cite{Farris}. We show how the fuzziness parameter $\kappa$ plays an analogous role to temperature in physical applications and may be used to break the symmetry of spatial patterns and to tune their smoothness.

As a further application to nonlinear dynamics, we show how the shape of limit-cycle oscillations in smooth nonlinear deterministic systems governed by systems of differential equations can be explicitly designed beforehand by using fuzzy sets as introduced in this article. We prove a general theorem that establishes the existence of a limit cycle with predefined shape for a wide variety of smooth dynamical systems and illustrate this theorem with examples. Our theorem can be useful in the mathematical modelling of, e.g. relaxation oscillators \cite{Ermentrout}, i.e. nonlinear oscillators whose periodic behavior strongly depart from sinusoidal oscillations. Typical examples of these oscillators, of biological interest, are provided by the van der Pol oscillator \cite{vanderpol}, the FitzHugh-Nagumo model for nerve cells \cite{Fitzhugh,Nagumo}, the Oregonator for the Belousov-Zhabotinsky reaction \cite{Fields,Belousov,Zhabotinsky}, the Morris-Lecar model for neuronal oscillations \cite{Morris}, genetic oscillators \cite{Guantes}, and modified versions of the Selkov model for glycolytic oscillations \cite{Segel,Goldbeter,Gonze,Westermark}.

%Our definitions are grounded on a theorem that we first state and prove for crisp sets and which we then proceed to generalize to fuzzy sets by exploiting the properties of the $\mathcal{B}_{\kappa}$-function.  

The outline of this article is as follows. In Section \ref{mathprembfun} we present mathematical structures that allow characteristic functions for crisp sets to be defined. These are the $\mathcal{B}$-function and its modifications. We then prove a theorem that yields a closed expression for the characteristic function of the union of a collection of sets. In Section \ref{fuzzytheo} we present our concise approach to fuzzy sets, based on the results of Section \ref{mathprembfun}. Finally, in Section \ref{dynamicdays} we prove a Theorem on the existence of limit-cycles with predesigned shape for a wide variety of smooth dynamical systems. Our result is both an application of the Poincar\'e-Bendixson theorem and of fuzzy sets as introduced in this article. We then discuss two specific examples (elliptic limit cycles and Cassini ovals) of the application of our theorem. 

\section{The $\mathcal{B}$-function and crisp sets} \label{mathprembfun}

In this section we introduce some functions that allow the characteristic functions of crisp sets to be expressed in an alternative, useful way. Let  $U$ denote the universe of discourse and let $x$ be an element. Let also $A \subseteq U$ be a set. We say that $A$ is a  'crisp' set if the membership of the element $x$ to $A$ is provided in terms of a bivalent condition in which $x$ either fully belongs or does not belong  to $A$. The characteristic (membership) function of a crisp set $A$, $\text{Set}(x; A): U\to \{0,1\}$, has thus the form \cite{Whitney1,Whitney2}
\begin{align}\label{charfu}
\begin{split}
\text{Set}\left(x; A\right)={\begin{cases} \text1 &{\text{if }}x \in A \\ 0 &{\text{if }}x \notin A  \end{cases}}
\end{split}
\end{align}
The complement set $A^c$ of the set $A$ is given by
\begin{equation}
\text{Set}\left(x; A^c\right)=1-\text{Set}\left(x; A\right) \label{charcomset}
\end{equation}
For the intersection of $n$ crisp sets $S^{(1)}$, $\ldots$, $S^{(n)}$, the characteristic function $\text{Set}\left(x; \cap_{j=1}^{n}S^{(j)} \right)$ is equal to the product of the characteristic functions of the individual sets
\begin{align}\label{interse}
\begin{split}
\text{Set}\left(x; \bigcap_{j=1}^{n}S^{(j)} \right)=\prod_{j=1}^{n}\text{Set}\left(x; S^{(j)}\right)
\end{split}
\end{align}

Finally, the union of $n$ sets $S^{(1)}$, $\ldots$, $S^{(n)}$ has characteristic function $\text{Set}\left(x; \bigcup_{j=1}^n S^{(j)}\right)$ given by the inclusion-exclusion principle \cite{Stanley} as
\begin{eqnarray}
&& \text{Set}\left(x; \bigcup_{j=1}^n S^{(j)}\right)  =\nonumber \\ &&  \sum_{k = 1}^{n} (-1)^{k+1} \sum_{1 \leq j_{1} < \cdots < j_{k} \leq n} \text{Set}\left(x; \bigcap_{m=1}^k S^{(j_{m})}\right)  \label{inextrape}
\end{eqnarray}

In this work, we shall consider $U=\mathbb{R}$.
The characteristic function, as defined by Eq. (\ref{charfu}), takes only values 0 and 1 and is compactly supported. This suggests that the latter can be defined in terms of a $\mathcal{B}$-function that we have used in a recent formulation of a universal map for cellular automata \cite{VGM2,VGM3} and subtitution systems \cite{VGM4}, since the $\mathcal{B}$-function also takes only a finite number of values $0,\pm 1/2,\pm 1$. The interest of this strategy lies in the fact that the $\mathcal{B}$-function can be easily embedded in the real numbers by means of its one-parameter $\kappa$-deformation, the $\mathcal{B}_{\kappa}$-function \cite{JPHYSA} that, as we shall see, leads to a new construction for fuzzy sets.  

For arbitrary $x, y \in \mathbb{R}$ the $\mathcal{B}$-function is defined as \cite{VGM1}
\begin{eqnarray}
\mathcal{B}(x,y)&=&\frac{1}{2}\left(\frac{x+y}{|x+y|}-\frac{x-y}{|x-y|}\right) \nonumber \\
&=&\frac{1}{2}\left(\text{sign}(x+y)-\text{sign}(x-y)\right) \nonumber \\
&=&{\begin{cases} \text{sign}(y)&{\text{if }}|x| < |y|\\ \text{sign}(y)/2 &{\text{if }}|x|=|y|, y\ne 0 \\0& {\text{otherwise}} \end{cases}} \label{d1}
\end{eqnarray}
The properties $\mathcal{B}(-x,y)=\mathcal{B}(x,y)$ and $\mathcal{B}(x,-y)=-\mathcal{B}(x,y)$ hold, i.e. the $\mathcal{B}$-function is an even function of its first argument and an odd function of its second argument. For positive $y$, this function has the form of a rectangular function of unit height in the interval $x<|y|$, with value $1/2$ at $|x|=y$ and $0$ otherwise \cite{JPHYSA}. For $y=0$ the function is zero everywhere. The $\mathcal{B}$-function is represented in Fig. \ref{boxen2} (top left) for $y=\epsilon>0$. Because of the $\mathcal{B}$-function being an odd function of its second argument, 
for negative $y$ it has the form of a rectangular well with values $-1$ in the interval $x<|y|$, $-1/2$ at $|x|=|y|$ and $0$ otherwise \cite{JPHYSA}. 

If $x\in \mathbb{Z}$ and $y \in \left.\left(0,\frac{1}{2}\right.\right]$, $y\in \mathbb{R}$, the $\mathcal{B}$-function, Eq. (\ref{d1}), reduces to the Kronecker delta of its first argument
\begin{equation}
\mathcal{B}(x,y)=\delta_{x0}={\begin{cases} 1&{\text{if }}x =0,  ~  y \in \left.\left(0,\frac{1}{2}\right.\right], y \in \mathbb{R} \\ 0 &{\text{if }}x\ne 0, ~ x\in \mathbb{Z}, y \in \left.\left(0,\frac{1}{2}\right.\right], y \in \mathbb{R} \end{cases}} \label{d1Kro}
\end{equation}
and, therefore, if $m, n \in \mathbb{Z}$, then $\mathcal{B}\left(m-n,\varepsilon \right)=\mathcal{B}\left(m-n,\frac{1}{2}\right)=\delta_{mn}$ for any $\varepsilon \in \mathbb{R}$ such that $\varepsilon \in \left.\left(0,\frac{1}{2}\right.\right]$.

The following is the \emph{splitting property} of the $\mathcal{B}$-function, and is valid for any $x, y, z \in \mathbb{R}$
\begin{equation}
\mathcal{B}(x,y+z)=\mathcal{B}(x+y,z)+\mathcal{B}(x-z,y) \label{splito} 
\end{equation}
as can be easily checked. For $n\in \mathbb{N}$, $x, \varepsilon \in \mathbb{R}$ we also have the \emph{block coalescence property}
\begin{eqnarray}
\sum_{k=0}^{n-1}\mathcal{B}(x-2k\varepsilon,\varepsilon)&=&\mathcal{B}(x-(n-1)\varepsilon,n\varepsilon) \label{induin} 
\end{eqnarray}
This is proved by induction. For $n=1$ the result is obviously valid. Let us assume it also valid for $n$ terms. Then, for $n+1$ terms we have
\begin{eqnarray}
&&\sum_{k=0}^{n}\mathcal{B}(x-2k\varepsilon,\varepsilon)=\mathcal{B}(x-2n\varepsilon,\varepsilon)+\sum_{k=0}^{n-1}\mathcal{B}(x-2k\varepsilon,\varepsilon) \nonumber \\
&&=\mathcal{B}(x-2n\varepsilon,\varepsilon)+\mathcal{B}(x-(n-1)\varepsilon,n\varepsilon) \nonumber \\
&&=\mathcal{B}(x-n\varepsilon,(n+1)\varepsilon)-\mathcal{B}(x-(n-1)\varepsilon,n\varepsilon)+\nonumber \\
&&\ +\mathcal{B}(x-(n-1)\varepsilon,n\varepsilon)=\mathcal{B}(x-n\varepsilon,(n+1)\varepsilon) \nonumber
\nonumber
\end{eqnarray} 
where the splitting property, Eq. (\ref{splito}) has been used. This establishes the validity of Eq. (\ref{induin}).

%\quad \text{(i.e)} \quad 
%\frac{1}{2}\left(\frac{x-r}{|x-r|}-\frac{x-q}{|x-q|}\right) =1
Any \emph{inequality} of the form $r<x<q$ involving the real numbers $r$ and $q$ ($r<q$) can thus be replaced by an \emph{identity}
\begin{equation}
\mathcal{B}\left(x-\frac{q+r}{2},\frac{q-r}{2}\right)=1  \label{interesting1}
\end{equation}
We note that, by definition, the $\mathcal{B}$-function, as defined above, returns $\pm 1/2$ at the borders, the sign depending on the second argument. In order to handle the values at the borders it proves convenient to introduce some other functions defined in terms of the $\mathcal{B}$-function. These are the functions $\mathcal{B}_{++}(x,y)$, $\mathcal{B}_{--}(x,y)$, $\mathcal{B}_{+-}(x,y)$ and $\mathcal{B}_{-+}(x,y)$ which we define as: 
\begin{eqnarray}
\mathcal{B}_{++}(x,y)&\equiv& \mathcal{B}(0,\mathcal{B}(x,y)) \nonumber \\
&=&{\begin{cases} 1&{\text{if }}x\in [-y,y], y>0\\-1&{\text{if }}x\in [y,-y], y<0\\0&{\text{if }}x\notin [-|y|,|y|] \end{cases}} \label{bclos} \\
\mathcal{B}_{--}(x,y)&\equiv& 2\mathcal{B}(x,y)-\mathcal{B}_{++}(x,y)\nonumber \\
&=&{\begin{cases} 1&{\text{if }}x\in (-y,y), y>0\\-1&{\text{if }}x\in (y,-y), y<0\\0&{\text{if }}x\notin (-|y|,|y|) \end{cases}} \label{bope} \\
\mathcal{B}_{-+}(x,y)&\equiv& \mathcal{B}\left(0,|x+y|\mathcal{B}(x,y)\right)\nonumber \\
&=&{\begin{cases} 1&{\text{if }}x\in (-y,y], y>0\\-1&{\text{if }}x\in [y,-y), y<0\\0&{\text{if }}x\notin (-|y|,|y|] \end{cases}} \label{bminplus} \\
\mathcal{B}_{+-}(x,y)&\equiv& \mathcal{B}\left(0,|x-y|\mathcal{B}(x,y)\right)\nonumber \\
&=&{\begin{cases} 1&{\text{if }}x\in [-y,y), y>0\\-1&{\text{if }}x\in (y,-y], y<0\\0&{\text{if }}x\notin [-|y|,|y|) \end{cases}} \label{bplusmin}
 \end{eqnarray} 
 
\begin{figure*} 
\centering\includegraphics[width=6.9in]{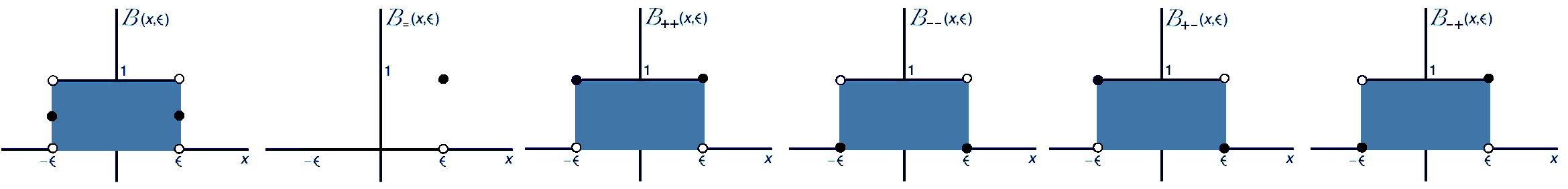}
%%% where xxxxxx name represents "figurename.eps"
\caption{\scriptsize{The functions $\mathcal{B}(x,y)$, $\mathcal{B}_{=}(x,y)$, $\mathcal{B}_{++}(x,y)$, $\mathcal{B}_{--}(x,y)$, $\mathcal{B}_{+-}(x,y)$ and $\mathcal{B}_{-+}(x,y)$, from left to right for $y=\epsilon >0$. }}
\label{boxen2}
\end{figure*}

We also note that thanks to the $\mathcal{B}$-function we can also generalize the Kronecker delta to arbitrary $x\in \mathbb{R}$ and $y \in \mathbb{R}$ in terms of the function
\begin{eqnarray}
\mathcal{B}_{=}(x,y)&\equiv& \mathcal{B}\left(\mathcal{B}(x,y)-\mathcal{B}(y,x),\frac{1}{2} \right)\nonumber \\
&=&  {\begin{cases} 1&{\text{if }}x =y \qquad x, y \in \mathbb{R} \\ 0 &{\text{if }} x\ne y \qquad  x, y \in \mathbb{R} \end{cases}} \label{d1KroEXT}
\end{eqnarray}

%Obviously, $\text{Set}\left(x; U\right)=1$. %Let $B \subseteq U$ be another set.  

We now discuss how the above functions can be used to describe arbitrary crisp sets over the real numbers. Any such set can be seen as a collection of isolated points and/or intervals and, therefore, its characteristic function can be expressed in terms of a superposition of the above functions (by virtue of the trichotomy property of the real numbers). For example, the characteristic function $\text{Set}(x; [a,b])$ of the closed interval $[a,b]$ ($a<b$) on the real line  can be expressed as
\begin{equation}
\text{Set}(x; [a,b])=\mathcal{B}_{++}\left(x-\frac{b+a}{2},\frac{b-a}{2} \right)
\end{equation}
In Table \ref{tableset} below, we provide the appropriate $\mathcal{B}$-function to describe, in each case, usual elementary subsets of $\mathbb{R}$.

\begin{table}[htp]
\caption{Elementary subsets of $\mathbb{R}$ and expressions of their characteristic functions in terms of the relevant $\mathcal{B}$-functions.}
\begin{center}
\begin{tabular}{c|c|c}
& Relevant  &  \\
& $\mathcal{B}$-function(s) & Example(s)  \\
\hline
isolated point & $\mathcal{B}_{=}(x,y)$ & $\text{Set}(x; \{a \})=\mathcal{B}_{=}\left(x,a \right)$\\
\hline
set of $n$  & $\mathcal{B}_{=}(x,y)$ & $\text{Set}(x; \{a_{1},\ldots,a_{n} \})$\\
 isolated points   & & $=\sum_{k=1}^{n} \mathcal{B}_{=}\left(x,a_{k} \right)$ \\
\hline
open  & $\mathcal{B}_{--}(x,y)$ & $\text{Set}(x; (a,b))$\\
interval & &$=\mathcal{B}_{--}\left(x-\frac{b+a}{2},\frac{b-a}{2} \right)$\\
\hline
closed  & $\mathcal{B}_{++}(x,y)$ & $\text{Set}(x; [a,b])$\\
interval & &$=\mathcal{B}_{++}\left(x-\frac{b+a}{2},\frac{b-a}{2} \right)$\\
\hline
 & $\mathcal{B}_{+-}(x,y)$  & $\text{Set}(x; [a,b))$\\
semiopen & or &$=\mathcal{B}_{+-}\left(x-\frac{b+a}{2},\frac{b-a}{2} \right)$\\
 intervals & $\mathcal{B}_{-+}(x,y)$ & $\text{Set}(x; (a,b])$\\
 & &$=\mathcal{B}_{-+}\left(x-\frac{b+a}{2},\frac{b-a}{2} \right)$\\
 \hline
\end{tabular}
\end{center}
\label{tableset}
\end{table}%

For example, the characteristic function of the closed rectangle $R$ in $\mathbb{R}^2$ with vertices at $(0,0)$, $(a,0)$, $(0,b)$ and $(a,b)$ is given by
\begin{eqnarray}
\text{Set}(x,y; R)&=&\text{Set}(x; [0,a])\text{Set}(y; [0,b]) \nonumber \\
&=&\mathcal{B}_{++}\left(x-\frac{a}{2},\frac{a}{2} \right)\mathcal{B}_{++}\left(y-\frac{b}{2},\frac{b}{2} \right)
\end{eqnarray}

As another example, the characteristic function of an open disk containing all those points $(x,y)$ in the plane $\mathbb{R}^{2}$ for which $x^{2}+y^{2}<R^{2}$ is
\begin{equation}
\text{Set}(x,y;\ x^{2}+y^{2}<R^{2})=\mathcal{B}_{--}\left(x^{2}+y^{2}, R^{2} \right) \label{opendisk}
\end{equation}

The above table can be used to systematically find the characteristic function of any other crisp set. We now prove a useful result which we shall need to establish the algebra of fuzzy sets.

\begin{theor} \label{Utheo} Let $x\in U$ be an element of the universe $U$, and let $S^{(1)}$,$\ldots$, $S^{(n)}$ denote a collection of $n$ sets of $U$. The characteristic functions $\text{\emph{Set}}\left(x; S_{= m}\right)$, $\text{\emph{Set}}\left(x; S_{\le m}\right)$ and $\text{\emph{Set}}\left(x; S_{> m}\right)$ of the elements $x$ that \emph{exactly} belong to $m$ sets, to \emph{no more than $m$} sets and to \emph{more than} $m$ sets of the collection are, respectively, given by 
\begin{eqnarray}
\text{\emph{Set}}\left(x; S_{= m}\right)&=&\mathcal{B}\left(m-\sum_{j=1}^{n}\text{\emph{Set}}\left(x; S^{(j)}\right), \varepsilon \right)   \label{subsets1} \\
\text{\emph{Set}}\left(x; S_{\le m}\right)&=&\mathcal{B}\left(\frac{m+1}{2}-\sum_{j=1}^{n}\text{\emph{Set}}\left(x; S^{(j)}\right), \frac{m}{2} \right)  \label{subsets2} \\
\text{\emph{Set}}\left(x; S_{> m}\right)&=&\mathcal{B}\left(\frac{n+m+1}{2}-\sum_{j=1}^{n}\text{\emph{Set}}\left(x; S^{(j)}\right), \frac{n-m}{2}\right)\nonumber \\ && \label{subsets3}
\end{eqnarray}
$\forall \varepsilon \in \left.\left(0,\frac{1}{2}\right.\right], \varepsilon \in \mathbb{R}$.
\end{theor}
%The proof of this result is given in \cite{VGMoutdated}, but it is reproduced here with some corrections, for the sake of completeness.
\begin{proof} 

%See \cite{VGMoutdated} for a proof.  \end{proof}
From Eq. (\ref{d1}), we find $\mathcal{B}\left(m-\sum_{j=1}^{n}\text{Set}\left(x; S^{(j)}\right), \varepsilon \right)=1$, for any $\varepsilon \in \left.\left(0,\frac{1}{2}\right.\right], \varepsilon \in \mathbb{R}$, if $-\varepsilon<m-\sum_{j=1}^{n}\text{Set}\left(x; S^{(j)}\right)<\varepsilon$. Thus $\sum_{j=1}^{n}\text{Set}\left(x; S^{(j)}\right)=m\in \mathbb{N}$, i.e. $x$ exactly belongs to $m$ sets. Otherwise, $\mathcal{B}\left(m-\sum_{j=1}^{n}\text{Set}\left(x; S^{(j)}\right), \frac{1}{2} \right)=0$. Thus, Eq. (\ref{subsets1}) follows.

%Eq. (\ref{subsets2}) is proved in a similar way. 

We now note that $\mathcal{B}\left(\frac{m+1}{2}-\sum_{j=1}^{n}\text{Set}\left(x; S^{(j)}\right), \frac{m}{2} \right)=1$ when
$-m<m+1-2\sum_{j=1}^{n}\text{Set}\left(x; S^{(j)}\right)<m$, and zero otherwise. Let  $k\equiv \sum_{j=1}^{n}\text{Set}\left(x; S^{(j)}\right)$. Then $0<1+2(m-k)<2m$, i.e. $1 \le k \le m$.  This means that $x$ is in the collection in less than $m$ sets. This proves Eq. (\ref{subsets2}).

If, now, $\mathcal{B}\left(\frac{n+m+1}{2}-\sum_{j=1}^{n}\text{Set}\left(x; S^{(j)}\right), \frac{n-m}{2}\right)=1$, then $-n+m<  n+m+1-2\sum_{j=1}^{n}\text{Set}\left(x; S^{(j)}\right) <n-m$. Let $k\equiv \sum_{j=1}^{n}\text{Set}\left(x; S^{(j)}\right)$. Thus, $0<  2(n-k)+1    <2(n-m)$ , i. e.  $m<k$.
Then $\sum_{j=1}^{n}\text{Set}\left(x; S^{(j)}\right)>m$, i.e $x$ belongs to more than $m$ sets of the collection. This proves Eq. (\ref{subsets3}). 
\end{proof}

%\\
% &=&\frac{1}{2}\left(\frac{2n+1-2\sum_{j=1}^{n}\text{\emph{Set}}\left(x; S^{(j)}\right)}{|2n+1-2\sum_{j=1}^{n}\text{\emph{Set}}\left(x; S^{(j)}\right)|}-\frac{1-2\sum_{j=1}^{n}\text{\emph{Set}}\left(x; S^{(j)}\right)}{|1-2\sum_{j=1}^{n}\text{\emph{Set}}\left(x; S^{(j)}\right)|}\right) \nonumber

\begin{corol} \label{inex} The characteristic function $ \text{\emph{Set}}\left(x; \bigcup_{j=1}^n S^{(j)}\right)$ of the union of $n$ sets $S^{(1)}$,$\ldots$, $S^{(n)}$ is
\begin{eqnarray}
 \text{\emph{Set}}\left(x; \bigcup_{j=1}^n S^{(j)}\right)  &=& \mathcal{B}\left(\frac{n+1}{2}-\sum_{j=1}^{n}\text{\emph{Set}}\left(x; S^{(j)}\right),\ \frac{n}{2} \right) \nonumber \\ && \label{mainU} 
\end{eqnarray}
\end{corol}

\begin{proof} The result follows by taking $m=n$ in Eq. (\ref{subsets2}) or $m=0$ in Eq. (\ref{subsets3}). 

%It is useful to note that this result can also be directly proved from the block coalescence property of the $\mathcal{B}$-function. From Eq. (\ref{induin}), by taking $m=\sum_{j=1}^{n}\text{Set}\left(x; S^{(j)}\right)$ and $\varepsilon=\frac{1}{2}$
%\begin{equation}
%\sum_{k=1}^{n}\mathcal{B}\left(\sum_{j=1}^{n}\text{Set}\left(x; S^{(j)}\right)-k, \frac{1}{2} \right)=\mathcal{B}\left(\frac{n+1}{2}-\sum_{j=1}^{n}\text{Set}\left(x; S^{(j)}\right),\frac{n}{2} \right)
%\end{equation}
%The left-hand-side of this latter expression is equal to $ \text{Set}\left(x; \bigcup_{j=1}^n S^{(j)}\right)$ since it evaluates to how many sets in the union does $x$ belong and returns one if it belongs to any of them and zero otherwise. This again proves the corollary.
\end{proof}

\begin{corol} \label{inex2} 
\begin{equation}
 \text{\emph{Set}}\left(x; S^{(j)}\right)  = \mathcal{B}\left(1-\text{\emph{Set}}\left(x; S^{(j)}\right),\ \frac{1}{2} \right) \label{mainUdos}
\end{equation}
\end{corol}
\begin{proof} Take $n=1$ in Eq. (\ref{mainU}) and note that $S^{(1)}$ may indeed refer to any set $S^{(j)}$ in the union since the $n$ different values of the integer label $j$ can be freely attributed. 
\end{proof}

Eq. (\ref{mainU}) is an alternative, equivalent expression of the inclusion-exclusion principle, Eq. (\ref{inextrape}).

Let us show an application of Theorem \ref{Utheo} to find the characteristic function of the set formed by the region in the plane bounded by a regular polygon of $p$ sides with apothem $L$. Let the point $(x_0,y_0)$ denote the center of the polygon in the plane $\mathbb{R}^2$. We begin by noting that the following inequality is valid for any point $(x,y)$ over a half-plane
\begin{equation}
L+(x-x_0) \cos \left(\frac{2k\pi}{p}+\varphi_{0}\right) 
+(y-y_0) \sin \left(\frac{2k\pi}{p}+\varphi_{0}\right)>0
\end{equation}
bounded by a line which is rotated an angle $2k\pi /p$, with $k=0,1,\ldots p-1$ ($\varphi_{0}$ sets the angle of the $k=0$ line with the $x$ axis), and which is at a distance $L$ of the origin. The above inequality can be translated into an identity by means of the $\mathcal{B}$-function as
\begin{eqnarray}
&&\mathcal{B}\left(0,  L +(x-x_0) \cos \left(\frac{2k\pi}{p}+\varphi_{0}\right) \right. \nonumber \\
&&\qquad \qquad \left.+(y-y_0) \sin \left(\frac{2k\pi}{p}+\varphi_{0}\right) \right)= 1 
\end{eqnarray}
and this corresponds to the characteristic function of the set of points belonging to the half-plane. There are $p$ such different lines, each obtained for the $p$ different values of $k$. All added together bound a polygonal domain, whose interior is the set we are seeking. If we add $N$ of these lines we have a line arrangement that breaks the plane into portions where the function takes different integer values. Any of these arrangements is given by
\begin{eqnarray}
&&\sum_{k=0}^{N-1}\mathcal{B}\left(0,  L +(x-x_0) \cos \left(\frac{2k\pi}{p}+\varphi_{0}\right) \right. \nonumber \\
&&\qquad \qquad \left.+(y-y_0) \sin \left(\frac{2k\pi}{p}+\varphi_{0}\right) \right) \label{stpolyint}
\end{eqnarray}
This function is plotted in Fig. \ref{pentagon} in the $x$-$y$ plane for $p=5$, $\varphi_0=x_0=y_0=0$, $L=2$ and $N=1$ (A), $N=2$ (B), $N=3$ (C), $N=4$ (D) and $N=5$ (E). The numerical integer values that the function takes on the plane are indicated. A polygon of $p$ sides (in this case a pentagon) is formed in the region where all half-planes intersect.
Therefore, a point within that region belongs exactly to $p$ sets (half-planes) in the collection. By taking $\varepsilon=1/2$ (for example), Eq. (\ref{subsets1}) automatically provides us with the characteristic function of all points in the plane that belong to the interior of the polygon
\begin{eqnarray}
&&\text{Polygon}(x, y; p, L)  \equiv \label{polyint} \\
 && \mathcal{B}\left(p-\sum_{k=0}^{p-1}\mathcal{B}\left(0,  L+(x-x_0) \cos \left(\frac{2k\pi}{p}+\varphi_{0}\right) \right.\right. \nonumber \\
 &&\qquad \qquad \qquad \qquad \left.\left.+(y-y_0) \sin \left(\frac{2k\pi}{p}+\varphi_{0}\right)\right),\frac{1}{2} \right) \nonumber
\end{eqnarray}
Therefore, the above function characterizes a regular polygon $P$ of $p$ sides with apothema $L$, rotated an angle $\varphi_0$ and centered at $(x_{0},y_{0})$. The function above returns $+1$ if $(x,y)\in \mathbb{R}^{2}$ is in the interior of $P$ and zero otherwise.

\begin{figure} 
\includegraphics[width=0.5\textwidth]{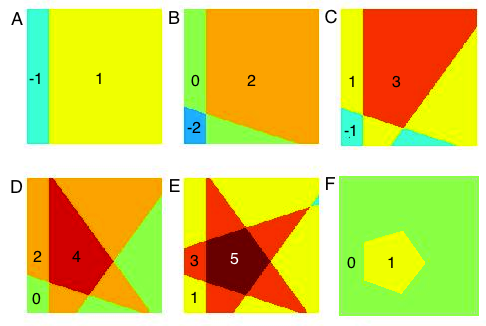}
\caption{\scriptsize{Plot of Eq. (\ref{stpolyint}) in the $x-y$ plane for $p=5$, $L=2$, $\varphi_{0}=0$ as the terms in the r.h.s. of Eq. (\ref{stpolyint}) are being added: 
values of $k=$0  (A), 1 (B), 2 (C), 3 (D), 4 (E). The numbers indicated clarify the integer values of the function in the color scheme used. (F) Plot of the function $\text{Poly}_{5}(x, y; 2, 0)$ given by Eq. (\ref{polyint}) on the plane.}} \label{pentagon}
\end{figure}

The function $\text{Polygon}(x, y; p, L)$ and the approach sketched above used to obtain it is illustrated in Fig. \ref{pentagon} for $p=5$. The 

At the bottom right of the figure (the panel within the box), the result Eq. (\ref{polyint}) is plotted for $p=5$, $L=2$, $(x_{0}, y_{0})=(4,4)$. The pentagon $\text{Poly}_{5}(4, 4, 2, 0)$ shown in the figure is given by Eq. (\ref{polyint}).

%Our expression for the characteristic function of the union of $n$ sets Eq. (\ref{mainU}) is thus an equivalent alternative to the one provided by the inclusion-exclusion principle, Eq. (\ref{inextrape}).

%The simplicity of our expression for the characteristic function of the union of $n$ sets Eq. (\ref{mainU}) is apparent and it admits a generalization to fuzzy sets as shown in the next section.

\section{Fuzzy sets: The $\kappa$-deformation of the $\mathcal{B}$-function} \label{fuzzytheo}

We now introduce $\mathcal{B}_{\kappa}$-function, $\kappa \in \mathbb{R}$, $\kappa \in (0,\infty)$, which acts as a one-parameter family of deformations of the $\mathcal{B}$-function, and is defined as
\begin{equation}
\mathcal{B}_{\kappa}(x,y)\equiv \frac{1}{2}\left( 
\tanh\left(\frac{x+y}{\kappa} \right)-\tanh\left(\frac{x-y}{\kappa} \right)
\right) \label{bkappa}
\end{equation}
In general, $0\le \left|\mathcal{B}_{\kappa}(x,y)\right| \le 1$. Indeed, for $y >0$ we have $0\le \mathcal{B}_{\kappa}(x,y) \le 1$. We also have $\mathcal{B}_{\kappa}(-x,y)=\mathcal{B}_{\kappa}(x,y)$ and $\mathcal{B}_{\kappa}(x,-y)=-\mathcal{B}_{\kappa}(x,y)$. Furthermore, if we allow $\kappa$ to be negative, $\mathcal{B}_{-\kappa}(x,y)=-\mathcal{B}_{\kappa}(x,y)$. For all finite values of the real variables $x$ and $y$, the $\mathcal{B}_{\kappa}$-function satisfies \cite{JPHYSA}:
\begin{eqnarray}
\lim_{\kappa \to \infty}\mathcal{B}_{\kappa}\left(x,\ y\right)&=& 0 \qquad  \qquad \
\lim_{\kappa \to \infty}\frac{\mathcal{B}_{\kappa}\left(x,\ y\right)}{\mathcal{B}_{\kappa}\left(0,\ y\right)}=1 \label{lim1} \\
\lim_{\kappa \to 0}\mathcal{B}_{\kappa}\left(x,\ y\right)&=& \mathcal{B}(x,y) \label{lim2}  \qquad \lim_{\kappa \to 0}\mathcal{B}_{\kappa}\left(0,\ y\right)= \text{sign}\ y \nonumber \\
&&
\end{eqnarray}
The function $\mathcal{B}_{\kappa}\left(x,\ y\right)$ is analytic both in $x$ and $y$.
Thus the $\mathcal{B}$-function is recovered from the $\mathcal{B}_{\kappa}$ function in the limit $\kappa \to 0$ (pointwise convergence). In the limit $\kappa \to \infty$ it is vanishingly small everywhere. However, note that, quite importantly
\begin{equation}
\frac{1}{2y}\int_{-\infty}^{\infty} dx \mathcal{B}_{\kappa}\left(x,\ y\right)=1
\end{equation}
a result that is independent of $\kappa$. Thus, if we take $y=\frac{1}{2}$, we can interpret $p(x)=\mathcal{B}_{\kappa}\left(x,\ \frac{1}{2} \right)$ as a probability distribution of $x \in \mathbb{R}$.
Because of all above properties, we refer to $\kappa$ as the \emph{fuzziness parameter}. When $\kappa \to 0$ the limit of crisp sets is obtained and for $\kappa$ increasing, the sets become more and more `vaguely' defined.

\begin{figure*} 
\centering\includegraphics[width=5.5in]{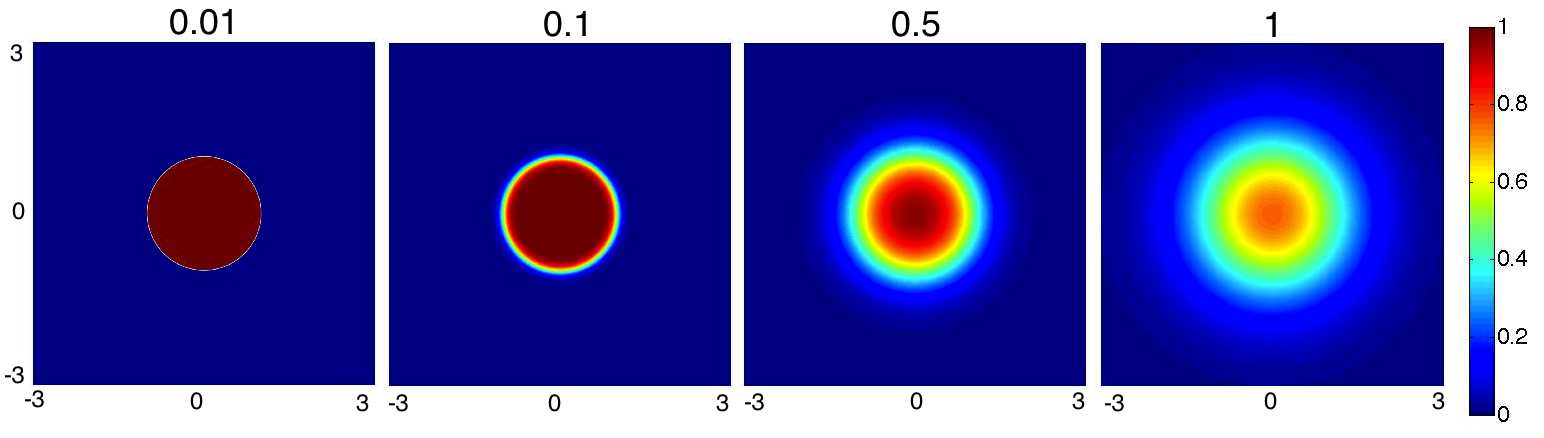}
%%% where xxxxxx name represents "figurename.eps"
\caption{\scriptsize{The fuzzy characteristic function $\text{Set}_{\kappa}(x,y;\ x^{2}+y^{2}<R^{2})$ corresponding to a disk of radius $R=1$ on the plane, for increasing values of the fuzziness parameter $\kappa$ indicated on the panels. }}
\label{fucir}
\end{figure*}

 The splitting and block-coalescence properties, Eqs. \ref{splito} and \ref{induin} also hold for the $\mathcal{B}_{\kappa}$-function \cite{JPHYSA}
\begin{eqnarray}
\mathcal{B}_{\kappa}(x,y+z)&=&\mathcal{B}_{\kappa}(x+y,z)+\mathcal{B}_{\kappa}(x-z,y) \label{splitk}\\
\sum_{k=0}^{n-1}\mathcal{B}_{\kappa}(x-2ky,y)&=&\mathcal{B}_{\kappa}(x-(n-1)y,ny) \label{induink}
\end{eqnarray}
for any $x, y, z \in \mathbb{R}$, $n\in \mathbb{N}\cup\{0 \}$ and for every $\kappa$.

The results in the previous section can now be fruitfully exploited to construct fuzzy sets. Let $S$ be a set for which a characteristic function $\text{Set}\left(x; S\right)$ can be defined in terms of the $\mathcal{B}_{++}$, $\mathcal{B}_{+-}$, $\mathcal{B}_{-+}$, $\mathcal{B}_{--}$ functions described in the previous section. The fuzzy version of the characteristic function $\text{Set}_{\kappa}\left(x; S\right)$ is simply obtained by replacing any occurrence of any of these functions by $\mathcal{B}_{\kappa}$ (note that in passing to the fuzzy characteristic function, whether the original set was open or closed becomes irrelevant). For example, the function
\begin{equation}
 \text{Set}_{\kappa}(x,y;\ x^{2}+y^{2}<R^{2})=\mathcal{B}_{\kappa}\left(x^{2}+y^{2}, R^{2} \right) \label{opendiska}
\end{equation}
constitutes the fuzzy version of Eq. (\ref{opendisk}). In Fig. \ref{fucir} this function is plotted in the plane for a disk with $R=1$ and the $\kappa$ values indicated over the panels.

We are thus considering the mapping $\text{Set}_{\kappa}\left(x; S\right): U \to [0,1]$ of the universe $U$ (containing the crisp set $S$) to the unit interval $[0,1]\subseteq \mathbb{R}$. We then describe the \emph{fuzzification} process as the operation of using this $\kappa$-dependent mapping to replace the mapping $\text{Set}\left(x; S\right): U \to \{0,1\}$ in any context where the latter appears. 

We note that if $x\in S$ and $x'\notin S$, we have $\text{Set}_{\kappa}\left(x; S\right) > 1/2 > \text{Set}_{\kappa}\left(x'; S\right) $. In the limit $\kappa \to 0$ this mapping becomes $\text{Set}_{\kappa \to 0}\left(x; S\right): U \to \{0,1/2,1\}$. Note that $\text{Set}_{\kappa \to 0}\left(x; S\right)$ is compactly supported and differs from the characteristic function of the crisp set $\text{Set}\left(x; S\right)$ only for those elements for which $\text{Set}_{\kappa \to 0}\left(x; S\right)=1/2$ (i.e. at the boundaries of the crisp set).

The complement $\overline{S}$ of a fuzzy set $S$ has characteristic function given by
\begin{equation}
\text{Set}_{\kappa}\left(x;\overline{S}\right)=1-\text{Set}_{\kappa}\left(x; S\right)
\end{equation}
The fuzzy union of $n$ fuzzy sets $S_{\kappa}^{j}$, $j=1,\ldots, n$  is given in terms of their $\kappa$-characteristic function by
\begin{equation}
\text{Set}_{\kappa}\left(x; \bigcup_{j=1}^n S^{(j)}_{\kappa}\right) = \mathcal{B}_{\kappa}\left(\frac{n+1}{2}-\sum_{j=1}^{n}\text{Set}_{\kappa} \left(x; S^{(j)}\right),\ \frac{n}{2} \right) \label{fuzzyunion}
\end{equation}
The validity of the latter equality is warranted by the fact that the $\mathcal{B}_{\kappa}$ function also obeys the block coalescence property, Eq. (\ref{induink}) and this property is all what is needed to prove Eq.(\ref{fuzzyunion}), thus generalizing Eq. (\ref{mainU}) to fuzzy sets (see the proof of Eq. (\ref{mainU}) above). The fuzzy intersection of $n$ sets is given by
\begin{equation}
\text{Set}_{\kappa}\left(x; \bigcap_{j=1}^n S^{(j)}_{\kappa}\right) = \prod_{j=1}^{n}\text{Set}_{\kappa}\left(x; S^{(j)}\right)
\end{equation}

Crisp sets can also be constructed out of fuzzy sets by using the tools introduced in the previous section. For example, starting from a fuzzy set $S_{\kappa}$, the crisp closed set with characteristic function
\begin{equation}
\mathcal{B}_{++}\left(\frac{1}{2}-\text{Set}_{\kappa}\left(x; S\right), \epsilon \right)
\end{equation}
is equal to one for all those elements that satisfy $\frac{1}{2}-\epsilon \le \text{Set}_{\kappa}\left(x; S\right)\le \frac{1}{2} + \epsilon$, i.e. those that lie on the `diffuse border' of the fuzzy set with characteristic function $\text{Set}_{\kappa}\left(x; S\right)$ up to tolerance $\epsilon$.

\section{Application: design of nonlinear time-series and symmetry-breaking patterns} \label{dynamicdays}

The switching function is important in electrical engineering in the analysis of power circuits \cite{Marouchos}. The basic (unipolar) switching function corresponds to a train of square pulses, each pulse having thickness $2\delta$ and periodicity $T$. Within the pulse, the switching function takes value '1' being equal to zero in the time span between pulses. The switching function is usually presented in the form of an (infinite) Fourier series (see, e.g. \cite{Marouchos}, p. 7). This function is represented in Fig. \ref{marouch} 
\begin{figure} 
\centering\includegraphics[width=3.2in]{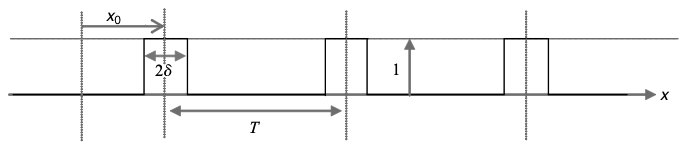}
%%% where xxxxxx name represents "figurename.eps"
\caption{\scriptsize{A sketch of the switching function. }}
\label{marouch}
\end{figure}

\begin{figure*} 
\centering\includegraphics[width=6.9in]{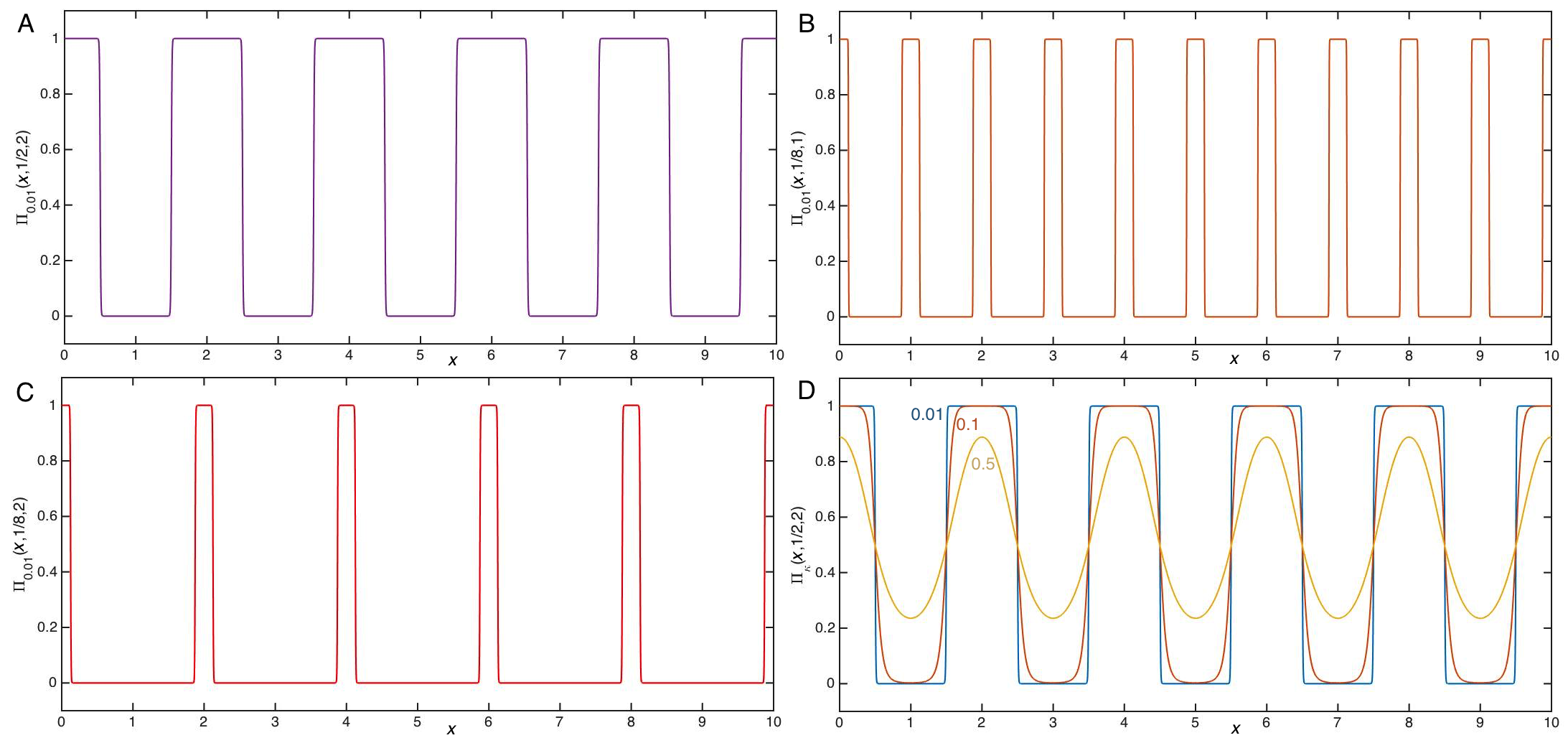}
%%% where xxxxxx name represents "figurename.eps"
\caption{\scriptsize{The $\kappa$-switching function $\Pi_{\kappa}(x,y; T)$ for specific values of the pulse thickness $y$, periodicity $T$ and smoothness parameter $\kappa$: $y=1/2$, $T=2$, $\kappa=0.01$ (A); $y=1/8$, $T=1$, $\kappa=0.01$ (B); $y=1/8$, $T=2$, $\kappa=0.01$ (C); $y=1/2$, $T=2$, and values of $\kappa$ indicated on the curves  (D).}}
\label{marouchi2}
\end{figure*}

By using the methods in this article, let us consider the universe $U=\mathbb{R}$ with $x\in \mathbb{R}$ being the time variable. We can then view the switching function as a characteristic (indicator) function of all those intervals in the real line for which $\ldots$,$-y-T<x<y-T$,  $-y<x<y$, $-y+T<x<y+T$, $\ldots$, being equal to 1 in those intervals and 0 everywhere else. Here $y=\delta$ denotes the pulse thickness. With the methods developed in Section \ref{mathprembfun}, it is easy to see that this characteristic function is then
\begin{equation}
\Pi(x,y;T)\equiv \mathcal{B}_{--}\left(\sin \frac{\pi x}{T}, \sin \frac{\pi y}{T} \right)
\end{equation}
from which we obtain the fuzzified version
\begin{eqnarray}
\Pi_{\kappa}(x,y; T)&\equiv& \mathcal{B}_{\kappa}\left(\sin \frac{\pi x}{T}, \sin \frac{\pi y}{T} \right) \label{ksw} \nonumber \\
&=& \frac{1}{2}\left( 
\tanh\left(\frac{\sin \frac{\pi x}{T}+\sin \frac{\pi y}{T}}{\kappa} \right)\right. \nonumber \\ && \left.-\tanh\left(\frac{\sin \frac{\pi x}{T}-\sin \frac{\pi y}{T}}{\kappa} \right)
\right)
\end{eqnarray}
The $\kappa$-switching function, Eq. (\ref{ksw}) is represented in Fig. \ref{marouchi2} for $\kappa=0.01$ (panels A, B, C) and other values (panel D) showing how an increased value of this parameter leads from a train of square pulses to a smoother periodic sine-like function. For vanishingly small $\kappa$, the switching function is a train of square pulses (Fig.\ref{marouchi2}A) having value 1 only in those intervals for which $-y-nT<x<y-nT$ with $n\in \mathbb{Z}$ and zero everywhere else. By adjusting $y$ to a different value, while keeping $T$ constant the thickness of the square pulses changes keeping the same periodicity (Fig.\ref{marouchi2}C). 

The $\kappa$-switching function can be used for the design of useful non-linear time series with arbitrary waveforms without need of expressing them in terms of Fourier components. Let $t$ denote the time variable and let $f(t)$ be an arbitrary function of time. Then, the function $f(t-t_0)\Pi_{\kappa}(t-t_0,\delta; T)$ is equal, for vanishing $\kappa$ to $f(t)$ in the intervals $-\delta-nT+t_0<t<\delta-nT+t_0$ with $n\in \mathbb{Z}$, and it is zero otherwise. For $\kappa \to \infty$, $f(t-t_0)\Pi_{\kappa}(t-t_0,\delta; T)$ is vanishingly small at every $t$ and a smooth function. However, we can change this limiting behavior as $\kappa \to \infty$ by using, Eq. (\ref{lim1}) and we can rather construct the signal
\begin{equation}
F(t)=f(t-t_0)\frac{\Pi_{\kappa}(t-t_0,\delta; T)}{\Pi_{\kappa}(0,\delta; T)}
\end{equation}
The behavior as $\kappa \to 0$ is the same as for the signal $f(t-t_0)\Pi_{\kappa}(t-t_0,\delta; T)$. In the limit $\kappa \to \infty$, because of Eq. (\ref{lim1}), we have $F(t) \to f(t-t_0)$. Thus, by tuning $\kappa$, one can construct a continuous signal composed of fragments of any arbitrary signal $f(t-t_0)$ which is periodically sampled at time intervals of length $T$ and whose samples are continuosly and analytically joined. Note that  $\Pi_{\kappa}(nT,\delta; T)=1=\Pi_{\kappa}(0,\delta; T)$ at all values $t=t'$ for which $t'-t_0=nT$, $n \in \mathbb{Z}$ and, therefore, $F(t_0+nT)=f(t_0+nT)$ regardless of the value of $\kappa$. Since the $\kappa$-switching function inherits the parity properties of the $\mathcal{B}_{\kappa}$-function, we have $\Pi_{\kappa}(x,y;T)+\Pi_{\kappa}(x,-y;T)=0$ a fact that can be used, e.g., to describe destructive interference phenomena involving complicated nonlinear time series as the one constructed above.

\begin{figure*} 
\centering\includegraphics[width=6.9in]{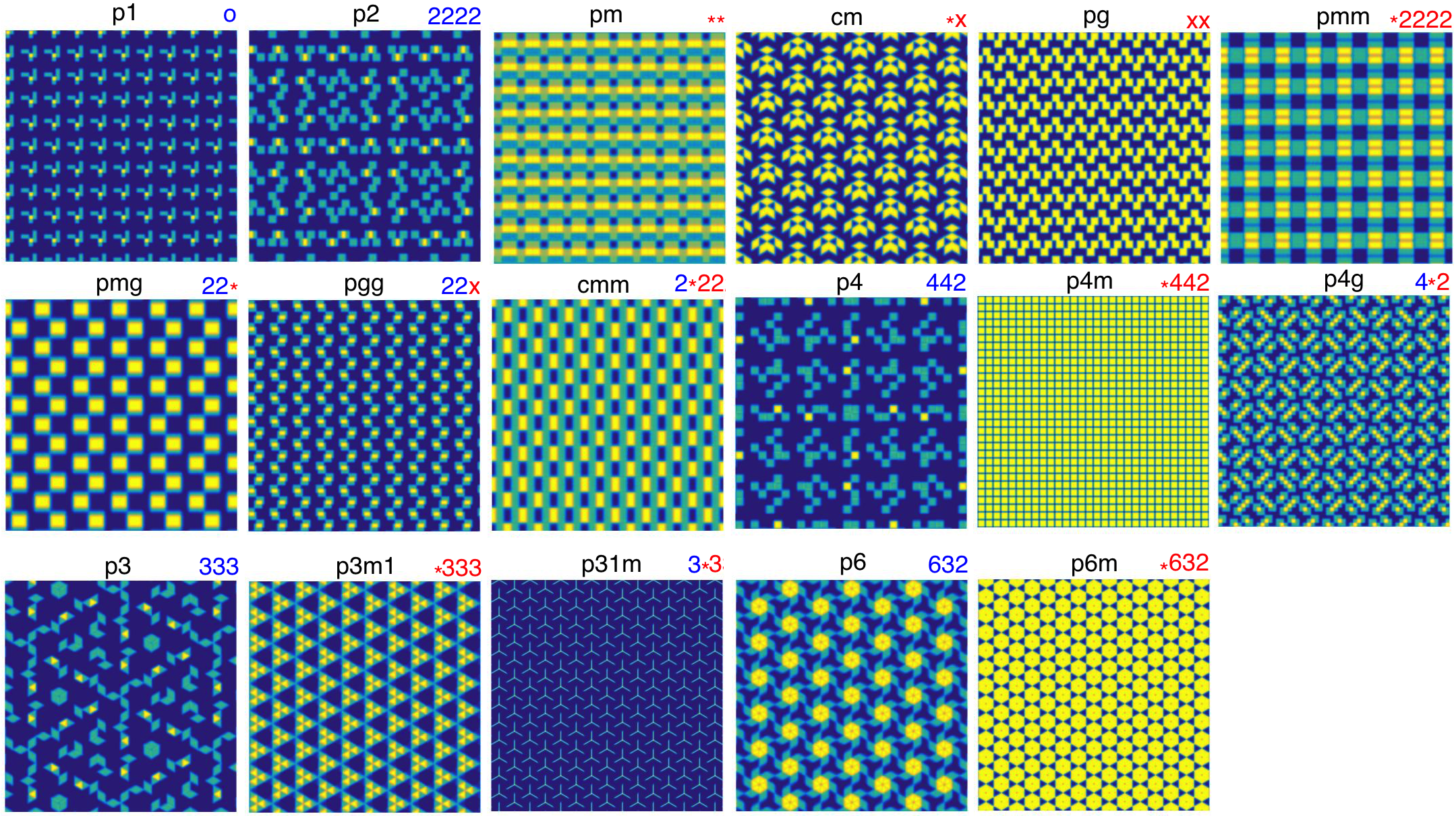}
%%% where xxxxxx name represents "figurename.eps"
\caption{\scriptsize{Spatial patterns obtained from Eq. (\ref{wallfun}) on the infinite plane. Shown is a window of size $30 \times 30$. The parameter values used to obtain the patterns are given in Table \ref{tableW}.}}
\label{wallp}
\end{figure*}

The above-defined $\kappa$-switching function can also be used to create spatial patterns that possess symmetries according to the 17 planar symmetry groups \cite{Farris,Conway2,Liu}. These symmetry groups are best understood from topological considerations and are most elegantly captured by Conway's orbifold notation \cite{Conway2,Conway1}. Let $p$ denote a centre of $p$-fold rotation corresponding to a cone point of the orbifold. By the crystallographic restriction theorem, $p=2,3,4$ or $6$ if one is to regularly fill the plane in terms of a repeated motif or cell (fundamental domain) so that the resulting arrangement is both ordered and periodic. Let $x$ and $y$ denote spatial variables on the cartesian plane and let $m$ and $n$ be parameters governing the spatial periodicity (translation symmetry) on these directions. We can then define the wallpaper function $W_{m,n,\kappa}^{(p)}(x,y)$ ($m, n\in \mathbb{R}$, $p\in \mathbb{N}$) as
\begin{eqnarray}
W_{m,n, \kappa}^{(p)}(x,y)&\equiv&\frac{1}{p}\sum_{k=1}^{p}\Pi_{\kappa}\left(\widetilde{x}_{k},a;m \right)\Pi_{\kappa}\left(\widetilde{y}_{k},b; n \right) \label{wallfun}
\end{eqnarray}
where
\begin{eqnarray}
\widetilde{x}_{k}&\equiv& (x-c_r)\cos\left(\frac{2\pi k }{p} \right)-(y-c_i)\sin\left(\frac{2\pi k}{p} \right)-a \nonumber \\
\widetilde{y}_{k}&\equiv& (x-d_r)\cos\left(\frac{2\pi (k+1)}{p} \right) \nonumber \\
&&-(y-d_i)\sin\left(\frac{2\pi (k+1)}{p} \right)-b 
\end{eqnarray}
with $a,b,c_r,c_i,d_r,d_i \in \mathbb{R}$ being parameters. Eq. (\ref{wallfun}) is expressed in terms of two $\kappa$-switching functions that govern the periodicity of the patterns in any spatial directions rotated on angles $2\pi k/p$, with $k=0,1\ldots, p-1$. We note that, rendering wallpaper patterns customarily involve a sum over a truncated Fourier series on different spatial directions \cite{Farris,Verberck}. In Eq. (\ref{wallfun}) this sum is always finite, and already allows to render patterns with all possible 17 symmetry groups by varying a set of a few parameters. In Fig. \ref{wallp} some of these patterns are shown for a window of $30 \times 30$ size on an infinite lattice. The wallpaper groups involved in the patterns are each specified by their crystallographic and orbifold notations, which are given in each case. In Table \ref{tableW}, the parameter values used in Eq. (\ref{wallfun}) to reproduce the patterns are given.

\begin{table}[htp]
\caption{Parameter values for the wallpaper functions $W_{m,n, \kappa}^{(p)}(x,y)$ in Fig. \ref{wallp}. Here, $c\equiv c_{r}+ic_{i}$ and $d\equiv d_{r}+id_{i}$.}
\begin{center}
\begin{tabular}{cc|cccccccc}
Wallpaper & Orbifold  					& $p$ &   $n$ & $m$ & $a$ & $b$ & $c$ & $d$ \\
group        & notation 					&         &        &         &        &         &        &       \\
\hline
p1             & {\color{blue} o} 				&  4    &     3   &   3   &   1    &  $\frac{1}{2}$   &   0 & $\frac{1}{2}$ \\ 
p2             & {\color{blue} 2222} 			&  4    &     3   &   4   &   1    &  1      		&   1 & $\frac{1}{2}$ \\ 
pm             & {\color{red} **}	 			&  4    &     3   &   3   &   2    &  2      		&   1 & 0 \\ 
cm             & {\color{red} *x}	 			&  6    &     4   &   4   &   1    &  1      		&   0 & 2 \\ 
pg             & {\color{red} xx}	 			&  4    &     3   &   3   &   1    &  1      		&   2 & $\frac{1}{2}$+2i \\ 
pmm         & {\color{red} *2222} 			&  4    &     4   &   2   &   1    &  2      		&   1 & 0 \\ 
pmg         & {\color{blue} 22}{\color{red}*} 		&  4    &     5   &   5   &   2    &  2      		&   $\frac{3}{4}$ & $\frac{3}{4}$+$\frac{5}{2}$i \\ 
pgg         & {\color{blue} 22}{\color{red}x} 		&  4    &     3   &   3   &   1    &  1      		&   1 & $\frac{3}{4}$+$\frac{9}{2}$i \\ 
cmm        & {\color{blue} 2}{\color{red}*22}		&  4    &     4   &   4   &   2    &  2      		&   $\frac{1}{2}$ & $\frac{3}{2}$ \\ 
p4             & {\color{blue} 442} 				&  4    &     3   &   4   &   1    &  1      		&   1 & $\frac{1}{2}$ \\
p4m             & {\color{red} *442} 			&  4    &     1   &   1   &   1    &  1      		&   0 & 0 \\
p4g        & {\color{blue} 4}{\color{red}*2}		&  4    &     3   &   2   &   1    &  1      		&   0 & $\frac{1}{2}$+$\frac{1}{2}$i \\ 
p3             & {\color{blue} 333} 				&  3    &     3   &   4   &   1    &  1      		&   0 & 0 \\
p3m1             & {\color{red} *333} 			&  3    &     3   &   3   &   2    &  1      		&   0 & 0 \\
p31m        & {\color{blue} 3}{\color{red}*3}		&  3    &     2   &   2   &   1    &  4      		&   0 & 0 \\ 
p6             & {\color{blue} 632} 				&  6    &     4   &   4   &   1    &  2      		&   0 & 0 \\
p6m             & {\color{red} *632} 			&  6    &     2   &   2   &   1    &  1      		&   0 & 0 \\
\end{tabular}
\end{center}
\label{tableW}
\end{table}%

In contrast with other methods, our formula for generating all wallpaper patterns consist of a finite number $p$ of terms. By varying $\kappa$, or the function parameters, or by adding linear combinations of the wallpaper functions, an infinite variety of wallpaper patterns with any planar symmetry group can be created. By introducing additional $\kappa$-switching functions in the product, patterns on higher dimensional spaces can be generated as well, e.g. the 230 space groups (this construction shall be presented elsewhere).
By selecting $p$ different to 2, 3, 4 or 6, ordered patterns that are not periodic, typical of quasicrystals \cite{Shechtman}, can be generated. For $p=5$, for example, since it is impossible to tile the plane with pentagons, quasicrystal patterns are formed. The latter completely fill the plane continuously, but lack any translational symmetry.
In Fig. \ref{quas} quasicrystal-like patterns are shown, as it is typical of wallpaper functions with $p=5$, for two different values of $\kappa$ indicated over the panels. These patterns are generated from Eq. (\ref{wallfun}) with $p=5$, $n=m=a=b=1$, $c_r=c_i=d_r=d_i=0$. At $\kappa$ large (left panel), the patterns are smoother while at $\kappa$ small (right panel) the patterns aproach a planar surface broken into pieces (tiles).  

\begin{figure} 
\centering\includegraphics[width=3.3in]{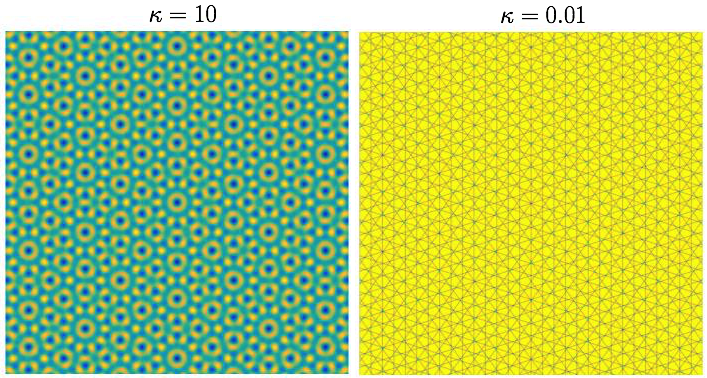}
%%% where xxxxxx name represents "figurename.eps"
\caption{\scriptsize{Spatial plot of the wallpaper function $W_{1,1, \kappa}^{(5)}(x,y)$ (i.e. $m=n=1$, $p=5$) for the values of $\kappa$ indicated over the panels. Other parameter values are $a=b=1$, $c_r=c_i=d_r=d_i=0$. Shown is a region of size $30\times 30$ of the infinite plane centered at the origin.}}
\label{quas}
\end{figure}

\begin{figure} 
\centering\includegraphics[width=3.3in]{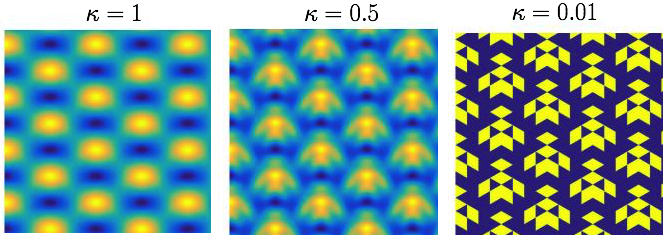}
%%% where xxxxxx name represents "figurename.eps"
\caption{\scriptsize{Spatial plot of the wallpaper function $W_{4,4, \kappa}^{(6)}(x,y)$ (i.e. $m=n=4$, $p=6$) for the values of $\kappa$ indicated over the panels. Other parameter values are $a=b=1$, $c_r=c_i=d_i=0$, $d_r=2$. Shown is a region of size $20\times 20$ of the infinite plane centered at the origin. }}
\label{symbr}
\end{figure}

The role of the parameter $\kappa$ is analogous to the one of temperature in physical systems. Symmetry breaking transitions can be induced by tuning this parameter and it is generally found that patterns with higher symmetry occur at $\kappa$ large, while less symmetric patterns are found at $\kappa$ lower.  In Fig. \ref{symbr} the spatial plot of the wallpaper function $W_{4,4, \kappa}^{(6)}(x,y)$ (i.e. $m=n=4$, $p=6$) for parameter values $a=b=1$, $c_r=c_i=d_i=0$, $d_r=2$ is shown in a region of the plane, for decreasing values of the parameter $\kappa$. It is observed that at $\kappa$ large ($\kappa=1$), the pattern has symmetry group cmm. However, as $\kappa$ is lowered ($\kappa=0.5$) one reflection axis is lost and the pattern collapses to the less symmetrical group cm. This less-symmetrical state remains as $\kappa$ is lowered more, the features of the pattern becoming sharper ($\kappa=0.1$).

Patterns with the frieze groups as symmetry groups extend orderly and periodically along one spatial dimension and can be directly obtained from the wallpaper functions defined above, by extracting a strip of the corresponding planar patterns. This is readily accomplished with the methods in Section \ref{mathprembfun} since extracting such a strip is analogous to multiplying by its characteristic function, which, for a strip centered at $(x,y)=(0,f)$ and with thickness $f$, is simply given by $\mathcal{B}\left(y-f,f\right)$. Therefore, we define the frieze function $F_{m,n}(x,y)$ ($m, n\in \mathbb{R}$) as
\begin{eqnarray}
F_{m,n,\kappa}^{(p)}(x,y)&\equiv&W_{m,n,\kappa}^{(p)}(x,y)\mathcal{B}\left(y-f,f\right) \label{wallfunF}
\end{eqnarray}
This `frieze' function allows patterns with all 7 frieze groups to be generated. In Fig. \ref{friezes}, examples of patterns with the symmetries of each different frieze group obtained from Eq. (\ref{wallfunF}) are shown, together with their cristallographic and Conway's orbifold notations. The parameter values used in Eq. (\ref{wallfunF}) to generate these patterns are indicated in Table \ref{tableW2}.

\begin{figure} 
\centering\includegraphics[width=3.3in]{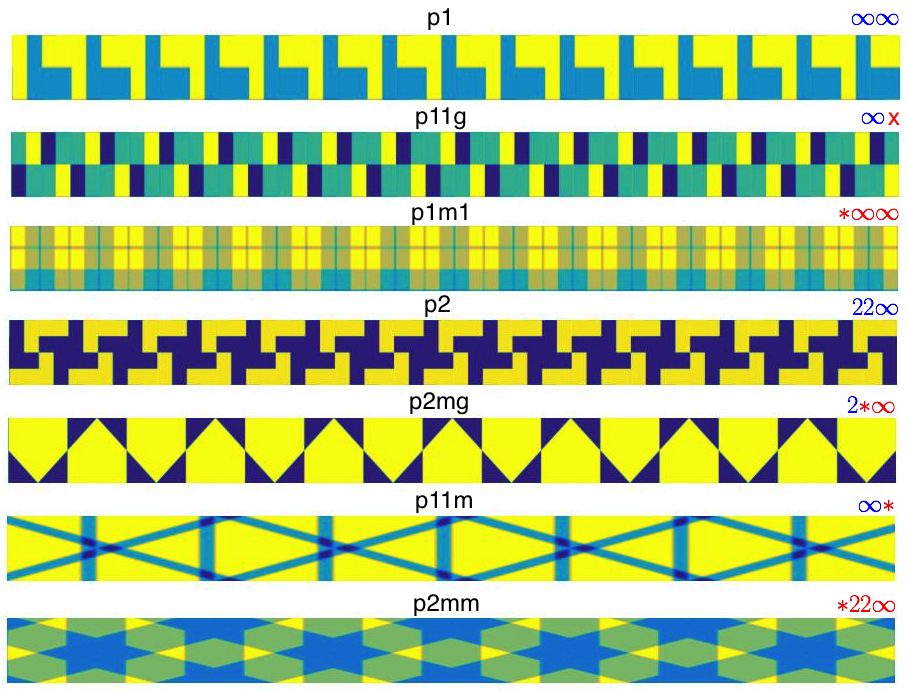}
%%% where xxxxxx name represents "figurename.eps"
\caption{\scriptsize{Spatial patterns obtained from Eq. (\ref{wallfunF}) on the plane. Shown is a window of size $30 \times f$. The value of $f$ and other parameter values  are given in Table \ref{tableW2}.}}
\label{friezes}
\end{figure}

\begin{table}[htp]
\caption{Parameter values for the frieze functions $F_{m,n, \kappa}^{(p)}(x,y)$ in Fig. \ref{friezes}. Here, $c_r=d_r=c_i=d_i=0$.}
\begin{center}
\begin{tabular}{cc|ccccccccc}
Frieze	& Orbifold  					& $p$ &   $n$ & $m$ & $a$ & $b$  & $f$ \\
group        & notation 					&         &          &         &        &          & \\
\hline
p1             & ${\color{blue} \infty \infty}$ 	        &  4    &     2   &   2   &   1      &  1.5  & 1\\ 
p11g          & ${\color{blue} \infty}${\color{red}x} &  4    &    1   &   2   &   0.5    &  1     & 1\\ 
p1m1          & ${\color{red} * \infty \infty}$  		&  4    &    1   &   2   &   1    &  1.5     & 1.5\\ 
p2             & ${\color{blue} 22\infty}$ 	        		&  4    &     2   &   2   &   1      &  0.5  & 2\\ 
p2mg          & ${\color{blue} 2}{\color{red}*\infty}$ &  6    &     2   &   2   &   1    &  1     & 1.15\\ 
p11m          & ${\color{blue} \infty}{\color{red}*}$  &  3    &     4   &   4   &   3.5    &  3.5     & 4.6\\ 
p2mm          & ${\color{red} *22\infty}$  		&  6    &     4   &   4   &   3    &  3     & 5\\
\end{tabular}
\end{center}
\label{tableW2}
\end{table}%

\section{Application: Shaping limit-cycle oscillations in nonlinear dynamical systems} \label{dynamicdays}

The concepts in Sections \ref{mathprembfun} and \ref{fuzzytheo} can also be applied to design limit-cycle oscillations far away from bifurcations of nonlinear dynamical systems, as we show in this section. %We present this application as a theorem and discuss some specific examples.

In biological systems, physiologically significant solutions are frequently periodic \cite{Cronin,GoldbeterBOOK} and their periodicity can be established by showing that there is an appropriate bounded set into which all the stable physiologically meaningful solutions enter and remain. A mathematical model of a physiological situation may take the form of a system of (nonlinear) ordinary differential equations
\begin{equation}
\frac{d\mathbf{y}}{dt}=\mathbf{f}(\mathbf{y}) \label{dynagen}
\end{equation}
Here $\mathbf{y}$ is a vector containing all dynamical variables and $t$ denotes the time variable. Given an arbitrary smooth dynamical system of the form of Eq. (\ref{dynagen}) to show whether it has stable periodic solutions is a hard, open important problem \cite{Cronin}. Most works in the field concentrate on specific dynamical systems, establishing the existence of limit cycles for them. 

In two dimensions, there is sometimes the possibility of using the Poincar\'e-Bendixson theorem  \cite{Cronin,Strogatz,Lefschetz}, or the criteria of Bendixson \cite{Lefschetz} and Dulac \cite{Dulac} to establish the existence (resp. non-existence) of periodic orbits \cite{Lloyd}. However, these results do not say anything on the shape of the resulting limit-cycle \cite{Dutta} (e.g., whether the oscillatory behavior is sinusoidal or more relaxation-like), a question of great interest in the mathematical modeling of biological systems. A very recent work \cite{Dutta} addresses this problem of finding the shape and size of the limit cycle and uses renormalization group techniques \cite{Banerjee} to reach conclusions in the specific case of the Selkov model for glycolytic oscillations.  

The problem of generally establishing the shape of a limit cycle for a given dynamical system is, obviously, harder than only proving the existence of the limit cycle. In an attempt to address features of this problem, we can consider the following closely related question that proceeds in the contrary direction: \emph{Suppose that we experimentally measure a limit cycle of a certain shape in a two-dimensional system; can then we formulate a mathematical model that is able to capture that experimentally observed shape?} 

We give an affirmative answer in the form of a theorem below. It is to be noted that, although we use characteristic functions of fuzzy sets (as introduced in this article), the system given by Eq. (\ref{dynagen})  is purely deterministic. The fuzzification process is here used to construct the basin of attraction of the limit cycle (of predetermined shape). The basin of attraction extends to the whole plane $\mathbb{R}^2$, but initial conditions nearer to the limit cycle have lower velocity. The characteristic function of the fuzzy set explicitly enters in the construction of the velocity field $\frac{d\mathbf{y}}{dt}$ in the whole plane, as a function of the position vector $\mathbf{r}$. The velocity field is designed so that if the point $\mathbf{y}$ belongs to the interior of the domain $D$ bounded by the limit cycle, the velocity is positive (for sufficiently low values of the positive real parameter $\kappa$), being negative if $\mathbf{r}$ is outside of $D$. In this way, thanks to explicit dependence on the appropriate characteristic function of the relevant fuzzy set, a trapping region associated to a smooth, differentiable field, is designed: The velocity field is defined in terms of a smooth differentiable characteristic function that is, in turn, a smooth indicator function of its own support. In this way, the fuzzification process allows the velocity field to be closely tied to the behavior of its support (whose size can be tuned by means of the parameter $\kappa$). In the limit $\kappa \to 0$, the limit cycle is continuously deformed to the asymptotically known shape (given by the closed curve $\partial D$). Otherwise, for $\kappa$ finite and small enough, the limit cycle is smaller in size but homeomorphic to $\partial D$. The freely adjustable control parameter $\kappa$ is here to be thought as directly related to the relevant, experimental control parameter governing the system dynamics. 

Our result is thus inspired in the concept of a trapping region as introduced in the Poincar\'e-Bendixson theorem. Thanks to the fuzzification method introduced in the previous Section, we can  relate the trapping region to the differentiable structure that it supports (the velocity field). In summary, our result extends the Poincar\'e-Bendixson theorem to the possibility of designing limit cycles with a predefined shape.

%Suppose there exists a closed classical set $A$ with characteristic function $\text{Set}(\mathbf{y}; A)$
%such that all physiologically meaningful solutions enter in $A$ and remain thereafter in $A$. If those solutions are stable. 

\begin{theor} \label{theoshape}
Let $D$ be a domain of the plane $\mathbb{R}^2$ containing a point called the origin, placed at position $\mathbf{r}_{0}=(x_0,y_0)$, and bounded by a closed curve $\partial D \in C^1$. Let a point $\mathbf{r}$ in the plane be given by coordinates $\mathbf{r}=(x_{0}+r\cos \theta,y_{0}+r\sin \theta)$ with $r$ being the radius and $\theta$ the angle and let $t$ denote the time variable and $\omega$ a natural frequency. Then, the  dynamical system
\begin{eqnarray}
\frac{dr}{dt}&=&\left( \text{\emph{Set}}_{\kappa}\left(\mathbf{r}-\mathbf{r}_0; D\right)-\frac{1}{2}  \right)r \label{radius} \\
\frac{d\theta}{dt}&=&\omega \label{angle}
\end{eqnarray}
has a stable limit cycle in the asymptotic regime $\kappa \sim 0$ with shape asymptotically given by $\partial D$. In the limit $\kappa\to \infty$ the origin is a stable fixed point and there is no limit cycle. The dynamics experiences a bifurcation between these two regimes at intermediate $\kappa$ values, as $\kappa$ is lowered below a critical value $\kappa_{c}$ given by the equation $\text{\emph{Set}}_{\kappa_c}\left(\textbf{\emph{0}}; D\right) = \frac{1}{2}$. If, furthermore,
\begin{equation}
\text{\emph{Set}}_{\kappa_c}\left(\textbf{\emph{0}}; D\right)=\max \text{\emph{Set}}_{\kappa_c}\left(\textbf{\emph{r}}-\textbf{\emph{r}}_{0}; D\right) \label{conditwo}
\end{equation}
then, the bifurcation at $\kappa_{c}$ is a supercritical Andronov-Hopf bifurcation.
\end{theor}

\begin{proof} The function $\text{Set}_{\kappa}\left(\mathbf{r}-\mathbf{r}_0; D\right)$ is continuous and differentiable for all finite non-vanishing $\kappa$ being also monotonically decreasing from the interior to the exterior of $D$. In the asymptotic regime $\kappa \sim 0$, there exist domains $D'$ and $D''$ in the plane such that $D'\subset D \subset D''$ and such that both $D'$ and $D''$ contain the origin, at position $\mathbf{r}_0$. Now, the radial velocity $dr/dt$ is positive on the boundary $\partial D'$ of $D'$, because $\text{Set}_{\kappa}\left(\textbf{r}; D'\right) >1/2$ while it is negative on the boundary $\partial D''$ of $D''$ where we have $\text{Set}_{\kappa}\left(\textbf{r}; D''\right) <1/2$. Thus, the difference set $D''\setminus D'$ is a connected set with interior boundary $\partial D'$ and exterior boundary $\partial D''$ and it does not contain the origin, being also a trapping region for the dynamics. Therefore, by the Poincar\'e-Bendixson theorem \cite{Strogatz,Bendixson,Teschl}, the trapping region contains a limit cycle (since there are no fixed points). The $\omega$-limit set thus approaches the shape $\partial D$  in the asymptotic regime $\kappa \sim 0$. 

Since we have, 
\begin{equation}
\frac{d\ \text{Set}_{\kappa}\left(\textbf{r}-\mathbf{r}_0; D\right)}{d\kappa} < 0
\end{equation}
with $\text{Set}_{\kappa}\left(\textbf{r}-\mathbf{r}_0; D\right) >0$, independently of the value $\mathbf{r}$, we find that, in the limit $\kappa \to \infty$, $\text{Set}_{\kappa}\left(\mathbf{r}-\mathbf{r}_0; D\right) \to 0$ in the whole plane, the radial velocity being negative everywhere. As a consequence, the origin attracts all trajectories in this case, being a stable fixed point. We note that since $D$ contains the origin, if we lower $\kappa$ from a high enough value where there is only this stable fixed point, we find that at the critical value $\kappa_c$ given by $\text{Set}_{\kappa}\left(\textbf{0}; D\right)=1/2$ the fixed point at the origin loses stability to the stable limit cycle which, for each $\kappa$ value is given by the equation $\text{Set}_\kappa\left(\mathbf{r}-\mathbf{r}_0; D\right)=1/2$. The limit cycle is homeomorphically deformed to the shape $\partial D$ as $\kappa \to 0$.

To prove the last part of the theorem,  note that if Eq. (\ref{conditwo}) is satisfied, we further have 
\begin{equation}
1/2=\text{Set}_{\kappa_c}\left(\textbf{0}; D\right)\ge \max \text{Set}_{\kappa}\left(\textbf{r}-\mathbf{r}_0; D\right) 
\end{equation}
for $\kappa > \kappa_{c}$. The equal sign only holds if $\kappa = \kappa_{c}$. Furthermore, $\text{Set}_{\kappa_c}\left(\textbf{0}; D\right)\ge \max \text{Set}_{\kappa_c}\left(\textbf{r}-\mathbf{r}_0; D\right)$ where the equal sign only holds if $\mathbf{r}=\mathbf{r}_0$. Therefore, the radial velocity is everywhere negative on the 
plane at criticality, save at the vicinity of the origin, which becomes a center. As $\kappa$  is then further lowered below the critical value $\kappa_{c}$, the origin loses stability through a supercritical Andronov-Hopf bifurcation \cite{Kuznetsov} to a tiny, stable limit cycle that emanates from the origin (by enclosing it). This limit cycle attracts all trajectories in phase space. 
\end{proof}

\noindent \emph{Remark 1}: If Eq. (\ref{conditwo}) is not fulfilled, a (homoclinic-like) global bifurcation generally takes place at bifurcation parameter $\kappa_{c}$ given by the expression $\text{Set}_{\kappa_c}\left(\textbf{\emph{0}}; D\right) = \frac{1}{2}$ instead of the (local) Andronov-Hopf bifurcation.\\

\noindent \emph{Remark 2}: If we wish to model an experimental situation where the limit cycle is found at increasing values of a control parameter $\alpha$ instead of decreasing ones, we could e.g. take $\kappa \sim 1/\alpha$ in the Theorem or introduce any other suitable parametrization.\\

\noindent \emph{Example}: Let $X$ and $Y$ represent the concentrations of adenosine diphosphate and fructose-6-phosphate, respectively, as in Selkov's model for glycolytic oscillations, and let us take $r=\sqrt{(X-X_0)^2+(Y-Y_0)^2}$ and $\theta=\arctan\frac{Y-Y_0}{X-X_0}$ in Eqs. (\ref{radius}) and (\ref{angle}), respectively. Here, the origin is located at $\mathbf{r}_{0}=(X_0,Y_0)=\left(b,  b/(c+b^2)\right)$ being given in terms of the parameters $b$ and $c$ of Selkov's model (see Eqs. (2.7) and (2.8) in \cite{Dutta}). Now, by using for $\partial D$ a mathematical approximate expression of a tilted and deformed ellipse fitting the limit-cycle shown in Fig. 6 in \cite{Dutta}, the characteristic function $\text{Set}\left(\mathbf{r}-\mathbf{r}_{0}; D\right)$ in Eq. (\ref{radius}) corresponds to the crisp open set $D$ of the interior bounded by the limit cycle $\partial D$. By the fuzzification method, we now replace that characteristic function by its fuzzy counterpart $\text{Set}_{\kappa}\left(\mathbf{r}-\mathbf{r}_{0}; D\right)$. By varying the control parameter $\kappa$, we can tune the size of the limit cycle and we would qualitatively capture with our model in Theorem 1 the  dynamics of Selkov's model of glycolytic oscillations for the whole parameter regime considered in \cite{Dutta}. The control parameter $\kappa$ here can be viewed as directly related to the parameter $c$ entering in Selkov's model (see Eq. (2.1) in \cite{Dutta}). \\

To better understand the example above providing illustrations of our theorem, let us study in detail a simpler abstract instance. We can take, e.g., the domain in the plane bounded by an ellipse with radii $a$ and $b$ given by the equation
\begin{eqnarray}
{\displaystyle \left({\frac {x}{a}}\right)^{2}+\left({\frac {y}{b}}\right)^{2}=1}  \label{elishape}
\end{eqnarray}
This ellipse is the shape $\partial D$ in our model. We have $\mathbf{r}_{0}=(0,0)$ and, therefore,  $x=r\cos\theta$ and $y=r\sin\theta$, and the domain $D$ enclosed by this curve is the open set with characteristic function
\begin{equation}
\text{Set}\left(\mathbf{r}; D\right)=\mathcal{B}_{--}\left(\left({\frac {r\cos \theta}{a}}\right)^{2}+\left({\frac {r\sin \theta}{b}}\right)^{2}, 1 \right)
\end{equation}
where use of the function $\mathcal{B}_{--}$ is made because we are dealing with the open set bounded by $\partial D$ (see Table \ref{tableset}). The fuzzification method then proceeds by replacing this characteristic function by its fuzzy counterpart, which amounts here to merely replace the $\mathcal{B}_{--}$ function by the $\mathcal{B}_{\kappa}$-function, as described in the previous section
\begin{equation}
\text{Set}_\kappa \left(\mathbf{r}; D\right)=\mathcal{B}_{\kappa}\left(\left({\frac {r\cos \theta}{a}}\right)^{2}+\left({\frac {r\sin \theta}{b}}\right)^{2}, 1 \right) \label{dynF}
\end{equation}
The set $D$ also contains the origin and, hence, $\text{Set}\left(\textbf{0}; D\right)=1$. Therefore, the theorem predicts that there is a stable fixed point at $\kappa \to \infty$, while a stable limit cycle with shape $\partial D$ given by Eq. (\ref{elishape}) is obtained in the limit $\kappa \to 0$. From the theorem, a bifurcation between both regimes occurs at the critical parameter value $\kappa_{c}$ given by $\text{Set}_{\kappa_c}\left(\textbf{0}; D\right)=\frac{1}{2}$, i.e. at $\kappa_c=1/\text{arctanh}(1/2)\approx 1.8205$. Furthermore, a simple calculation also shows that $\text{Set}_\kappa \left(\mathbf{0}; D\right) \ge \text{Set}_\kappa \left(\mathbf{r}; D\right)$  $\forall \mathbf{r}$. Therefore, Eq. (\ref{conditwo}) is also satisfied and the theorem predicts that the bifurcation is a supercritical Antonov-Hopf bifurcation .

\begin{figure*} 
\centering\includegraphics[width=6.75in]{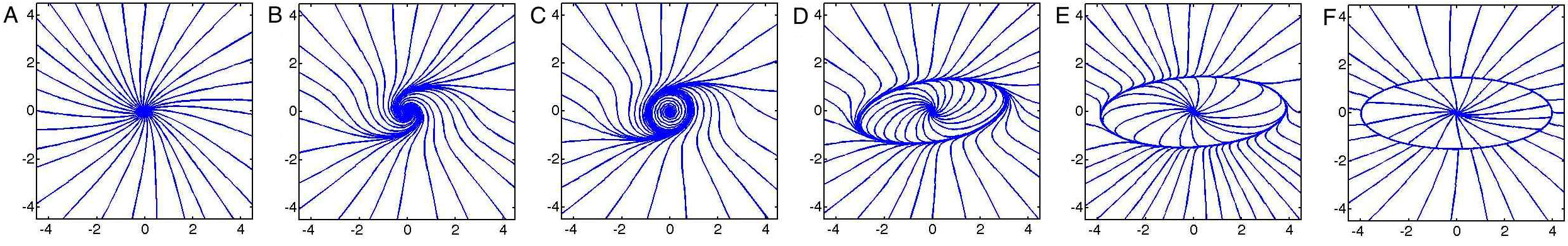}
%%% where xxxxxx name represents "figurename.eps"
\caption{\scriptsize{Phase portraits of the dynamical system given by Eqs.(\ref{radius}) and (\ref{angle}) with $\text{Set}_\kappa \left(\mathbf{r}; D\right)$ given by Eq. (\ref{dynF}) for the fuzziness parameter $\kappa$ values 5 (A), 1.9 (B), 1.8 (C), 1.5 (D), 1 (E), 0.1 (F). Other parameter values are $a=4$, $b=1.5$ and $\omega=0.05$. }}
\label{cicloss}
\end{figure*}

In Fig. \ref{cicloss} the phase portraits of the dynamical system, Eqs. (\ref{radius}) and (\ref{angle}), with $\text{Set}_\kappa \left(\mathbf{r}; D\right)$ given by Eq. (\ref{dynF}), are obtained by integrating both forward and backward in time starting from several different initial conditions on the plane and plotted for distinct $\kappa$ values. The origin is a stable fixed point/focus (panels A and B) and loses stability to a circle-shaped limit cycle at a supercritical Andronov-Hopf bifurcation found at $\kappa=\kappa_c$ with $\kappa_c=1/\text{arctanh}(1/2)\approx 1.8205$ (panel C) as predicted by the theorem. As the control parameter $\kappa$ is further decreased, the limit circle gradually deforms to an ellipse (panels D and E), so that at $\kappa=0.1$ already the border is very close to the limiting shape $\partial D$ which, in Fig. \ref{cicloss} corresponds to an ellipse with radii $a=4$ and $b=1.5$.  

%We note that the origin can be placed at any point $(x_0, y_0)$ within the ellipse and the theorem would also hold. This feature can be used in the modeling of relaxation oscillators with strongly asymmetric limit cycles. 

\begin{figure} 
\centering\includegraphics[width=2.5in]{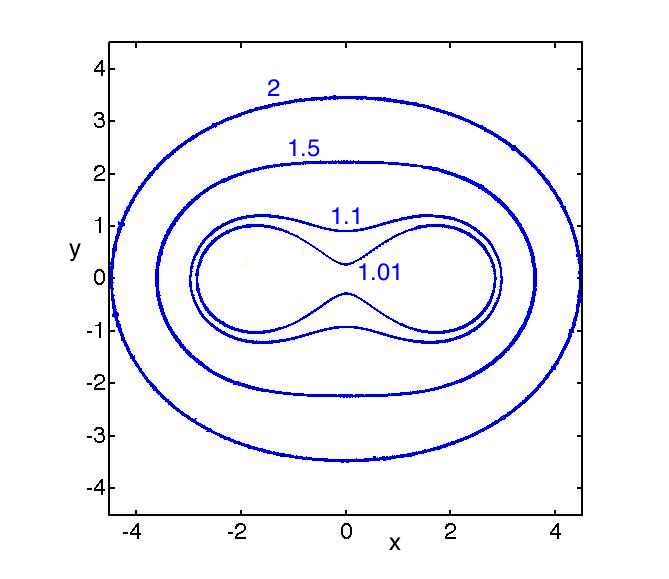}
%%% where xxxxxx name represents "figurename.eps"
\caption{\scriptsize{Limit cycles obtained in the long time limit for the dynamical system given by Eqs.(\ref{radius}) and (\ref{angle}) with $\text{Set}_\kappa \left(\mathbf{r}; D\right)$ given by Eq. (\ref{Cassik}) for $\kappa=0.1$, $a=1$, $\omega=0.02$ and four distinct values of $b$ indicated on the curves. }}
\label{Cassicy}
\end{figure}

We can still give another example of application of our theorem. Let us now consider the mathematical expressions for by Cassini ovals, which are defined by the set of points in the plane such that the product of the distances to two fixed points is constant. Cassini ovals are polynomial lemniscates of degree 2 given by the equation
\begin{equation}
r^{4}-2a^{2}r^{2}\cos 2\theta+a^{4}=b^{4} \label{Cass}
\end{equation}
where $a$ and $b$ are parameters. If $b<a$ the curve consists of two disconnected loops. For $b=a$ the curve becomes a Bernoulli lemmniscate with a double point at the origin \cite{Lawden,Basset}. When $b>a$ the curve is a single loop enclosing the origin and, in this case, our theorem applies. Let us then consider $b>a$, with $\partial D$ being given by Eq. (\ref{Cass}). The domain $D$ is, therefore, given by the inequality
\begin{equation}
r^{4}-2a^{2}r^{2}\cos 2\theta+a^{4}<b^{4} \label{Cass}
\end{equation}
Again, we take $\mathbf{r}_{0}=(0,0)$. The above open set is thus specified by the Boolean characteristic function
\begin{equation}
\text{Set}\left(\mathbf{r}; D\right)=\mathcal{B}_{--}\left(r^{4}-2a^{2}r^{2}\cos 2\theta+a^{4}, b^{4} \right)
\end{equation}
from which the fuzzy version is directly obtained as
\begin{equation}
\text{Set}_{\kappa}\left(\mathbf{r}; D\right)=\mathcal{B}_{\kappa}\left(r^{4}-2a^{2}r^{2}\cos 2\theta+a^{4}, b^{4} \right) \label{Cassik}
\end{equation}

In Fig. \ref{Cassicy} the $\omega$-limit sets of the flow given by Eqs.(\ref{radius}) and (\ref{angle}) with $\text{Set}_\kappa \left(\mathbf{r}; D\right)$ given by Eq. (\ref{Cassik}) are plotted in the plane for $\kappa=0.1$, $a=1$ and the values of $b$ indicated on the curves. These are all limit cycles and we observe that for $a < b < a\sqrt{2}$ the limit cycle has the shape of a peanut and for $a\sqrt{2} <2$ it bounds a convex domain. Thus, the limit cycles coincide with the curves given by Eq. (\ref{Cass}) in the asymptotic regime $\kappa$ small, as predicted by our theorem.

If we define the complex number $W\equiv \left|W\right|e^{i\theta}$, with $r=|W|^2$, we have
\begin{eqnarray}
\frac{dW}{dt}&=&\frac{d\left|W\right|}{dt}e^{i\theta}+i\left|W\right|e^{i\theta}\frac{d\theta}{dt} =\frac{e^{i\theta}}{2\left|W\right|} \frac{dr}{dt}+i\omega W \nonumber \\
&=&\left(\frac{1}{2\left|W\right|^2} \frac{dr}{dt}+i\omega\right)W=\left(\frac{1}{2r} \frac{dr}{dt}+i\omega \right)W\nonumber\\
&=&\frac{1}{2}\left(\text{Set}_{\kappa}\left(\mathbf{r}; D\right)-\frac{1}{2}+2\omega i      \right)W
\end{eqnarray}
and, therefore, by defining $t'\equiv t/2$, $\omega'\equiv 2\omega$ and dropping the tildes we finally obtain
\begin{equation}
\frac{dW}{dt}=\left(\text{Set}_{\kappa}\left(\mathbf{r}; D\right)-\frac{1}{2}+i\omega       \right)W \label{STG}
\end{equation}
with $\mathbf{r}=\left(\frac{W+W^*}{2}-x_0,\frac{W-W^*}{2i}-y_0 \right)$ where $W^*$ denotes the complex conjugate of $W$. This equation can be viewed as a generalization of the Stuart-Landau equation \cite{KuramotoBOOK,PRK,contemphys} to non-sinusoidal limit-cycle oscillators. Note that, because of the $W^*$ dependence, Eq. (\ref{STG}) is not generally invariant under the phase transformation $W\to We^{i\xi}$ where $\xi$ is an arbitrary constant phase: Phase invariance is broken as a consequence of the anharmonicity of the oscillations.

%Having established Theorem \ref{theoshape} on the existence of a limit cycle, we can now show that the Stuart-Landau equation is a particular case of Eq. (\ref{STG}) close to a Hopf bifurcation when there is also phase invariance. Indeed if phase invariance holds \cite{contemphys} then, necessarily,  $\text{Set}_{\kappa}\left(\mathbf{r}; D\right)=\text{Set}_{\kappa}\left(r; D\right)=\text{Set}_{\kappa}\left(\left|W\right|^2; D\right)$. Furthermore, close to a Hopf bifurcation, this latter expression is an analytic function of the radius $r=\left|W\right|^2$, and, hence, for $|W|$ sufficiently small, it admits a McLaurin series as
%\begin{equation}
%\text{Set}_{\kappa}\left(\left|W\right|^2; D\right)=\sum_{j=0}^{\infty}\rho_{j}\left|W\right|^{2j}
%\end{equation}
%where the coefficients $\rho_j$ are obtained as
%\begin{equation}
%\rho_{j}=\frac{1}{j!}\left.\frac{d^{j}}{dr^{j}}\text{Set}_{\kappa}\left(r; D\right)\right|_{r=0}
%\end{equation}

\section{Conclusion} \label{conclu}

In this work we have presented a new approach to fuzzy sets in terms of new definitions of the characteristic functions on sets and operations on them, and we have presented several nontrivial applications of our theory. The first of them concerns the stablishment of a periodic switching function that is easy and fast to compute and that can be used to design a wide variety of nonlinear time series ranging from trains of square pulses to smooth waves. When this switching function operates on spatial variables, we have shown how it can be used to systematically generate patterns with symmetries according to all planar frieze and wallpaper groups, as well as quasicrystalline patterns. The fuzzification parameter plays an analogous role to temperature in physical applications regarding phase transitions: patterns with lower symmetry are obtained at lower values of $\kappa$. These results may be interesting in diverse fields as the generation of aesthetic patterns through mathematical formulae, electrical engineering, condensed matter physics (phase transitions), etc.  

Another important application of our theory has been the establishment of a theorem on the shaping of limit-cycle oscillators in smooth nonlinear dynamical systems. The result may be of interest in the mathematical modeling of nonlinear oscillatory phenomena of biological interest, as those involving relaxation oscillators \cite{Ermentrout} (in which the shape of the limit cycle strongly departs from a circle). We have used the characteristic function of a fuzzy set to model the velocity field in the whole plane. The loci where this fuzzy characteristic function is equal to $1/2$ are connected forming a closed curve, and correspond to points with vanishing radial velocity. The fuzzy characteristic function allows the velocity to vary continuously on the plane, yielding a thoroughly differentiable flow. It is important to clarify that the 'fuzziness'  introduced in the statement of this result has nothing to do with probability and/or noise: the limit cycles here obtained \emph{always} constitute the sharp boundaries of an open crisp set corresponding to the locations enclosed by the limit cycle. However, while the characteristic function of a crisp set takes only values 0 and 1 and is, therefore, useless to define any differentiable structure, the \emph{continuous range} of the fuzzy characteristic function allows to smoothly merge together the open crisp set of locations with the velocity field. This strategy has the major advantage of yielding a limit cycle with tunable size and shape which is directly controlled by the only adjustable parameter $\kappa$. We have briefly sketched how our result relates to very recent literature on oscillators of biological interest \cite{Dutta}, although it is difficult at this stage to establish a closer relationship because our approach follows an opposite direction: Instead of trying to determine the shape of the limit cycle of an specific dynamical system, we take the asymptotic shape as given (e.g. by experimental measurement) and use this information in looking for a mathematical model that is able to describe the dynamics of a whole bifurcation scenario as the control parameter is varied. We have discussed in detail two specific, albeit abstract, examples involving elliptic limit cycles and attracting Cassini ovals. 

The approach to fuzzy sets presented makes use of the one-parameter family of functions called $\mathcal{B}_{\kappa}$-functions that we have recently introduced \cite{JPHYSA}. The fuzziness parameter $\kappa$ governs how much a smooth set departs from a `crisp' set. We have given expressions for the negation, union and intersection of fuzzy sets that are all different to the original formulation of fuzzy sets by Zadeh \cite{Zadeh}. Computation with our expressions is also straightforward and does not make use of the maximum/minimum operators for the union/intersection of fuzzy sets, constituting an alternative to Zadeh's theory. We note that Zadeh's theory can be reproduced through the methods in Section \ref{mathprembfun} since the maximum and minimum operators acting on characteristic functions $\chi_{A}\equiv \text{Set}(x; A)$, $\chi_{B}\equiv \text{Set}(x; B)$ of two sets $A$ and $B$ can be defined as
\begin{eqnarray}
&&\max\left(\chi_{A},\chi_{B}\right)=\frac{\chi_{A}+\chi_{B}}{2}+\frac{\chi_{A}-\chi_{B}}{2}\mathcal{B}\left(0,\chi_{A}-\chi_{B}\right) \nonumber \\ && \label{maxi} \\
&&\min\left(\chi_{A},\chi_{B}\right)=\frac{\chi_{A}+\chi_{B}}{2}-\frac{\chi_{A}-\chi_{B}}{2}\mathcal{B}\left(0,\chi_{A}-\chi_{B}\right) \nonumber \\ &&\label{mini}
\end{eqnarray}
We note, however, that although these operators lead to continuous characteristic functions these have not continuous derivatives. On the contrary, the fuzzification method proposed in this manuscript to obtain fuzzy sets out of crisp ones, lead to infinitely differentiable characteristic functions since these use all the $\mathcal{B}_{\kappa}$-function as building block and $\mathcal{B}_{\kappa}(x,y)$ is infinitely differentiable over $x$, $y$ and $\kappa$ (excluding $\kappa=0$).

With the help of fuzzy sets and the theorem on nonlinear oscillators proved here, we have established a differential equation for nonlinear oscillators in the complex plane, Eq. (\ref{STG}), in terms of a complex order parameter $W$. This equation can be viewed as a generalization of the Stuart-Landau equation  \cite{KuramotoBOOK,PRK,contemphys} to non-sinusoidal oscillators. Note that a Hilbert transform can be performed on periodic experimental signals bringing the dynamical system to a phase and amplitude representation \cite{PRK}. The empirical analysis of the shape of the oscillations as a function of the experimental control parameter, when brought in connection with the general model Eq. (\ref{STG}), may help to get insight in complex nonlinear phenomena. For example, an ensemble of non-sinusoidal oscillators, each individually described by Eq. (\ref{STG}) can be diffusively coupled so that a generalized Ginzburg-Landau equation is obtained. In this way, the impact on the spatiotemporal dynamics of the anharmonicity of the oscillations can be systematically studied.  

Although it constitutes a \emph{local} model at each point in phase space, Eq. (\ref{STG}) establishes a way in which \emph{global} information of the phase space affects the local dynamics as the control parameter $\kappa$ is being varied. Thus fuzzy sets, as presented here, allow global properties of the phase space of deterministic dynamical system to be addressed.

%\section*{Data accessibility} This work does not have any experimental data.
%\section*{Competing interests} We have no competing interests.
%\section*{Authors' contributions} VGM is the sole author of this manuscript and carried out the whole work contained in it.
%\section*{Acknowledgements} Constructive comments by anonymous referees are gratefully acknowledged.
%\section*{Funding statement} This work received no funding.

\end{document}